\documentclass[12pt]{article}

\usepackage{url}
\usepackage{bbm}
\usepackage{mathtools}
\usepackage{amssymb}
\usepackage{amsthm}
\usepackage{empheq}
\usepackage{latexsym}
\usepackage{eurosym}
\usepackage{dsfont}
\usepackage{appendix}
\usepackage{color} 
\usepackage[unicode]{hyperref}
\usepackage{frcursive}
\usepackage[utf8]{inputenc}
\usepackage[T1]{fontenc}
\usepackage{geometry}
\usepackage{multirow}
\usepackage{todonotes}
\usepackage{lmodern}
\usepackage{anyfontsize}
\usepackage{pgfplots}
\usepackage{stmaryrd}
\usepackage{float}

\usepackage{graphicx}
\usepackage{caption}
\usepackage{subcaption}
\usepackage{enumerate}
\usepackage{breakcites}
\graphicspath{ {images/} }

\pgfplotsset{compat=newest}

\definecolor{red}{rgb}{0.7,0.15,0.15}
\definecolor{green}{rgb}{0,0.5,0}
\definecolor{blue}{rgb}{0,0,0.7}
\hypersetup{colorlinks, linkcolor={red},citecolor={green}, urlcolor={blue}}
            
\makeatletter \@addtoreset{equation}{section}

\newtheorem{theorem}{Theorem}
\newtheorem{theorem2}{Theorem}[section]

\newtheorem*{assumption_sig*}{Assumption $\Ac_\sigma$}
\newtheorem*{assumption_lamb*}{Assumption $\Ac^\lambda$}
\newtheorem*{assumption_max*}{Assumption $\Ac^\text{max}$}
\newtheorem*{assumption_fix*}{Assumption $\Ac^\text{fix}$}
\newtheorem{corollary}[theorem2]{Corollary}

\newtheorem{proposition}[theorem2]{Proposition}

\newtheorem{remark}[theorem2]{Remark}
\newcommand{\comment}[1]{}
\DeclareUnicodeCharacter{014D}{\=o}
\def \C{\mathbb{C}}

\def \E{\mathbb{E}}
\def \F{\mathbb{F}}

\def \N{\mathbb{N}}

\def \P{\mathbb{P}}

\def \R{\mathbb{R}}

\def\Ac{{\cal A}}

\def\Cc{{\cal C}}

\def\Fc{{\cal F}}

\def\Lc{{\cal L}}
\def\Mc{{\cal M}}
\def\Nc{{\cal N}}

\title{A Bayesian viewpoint on the price formation process}
\author{Joffrey {\sc Derchu}\footnote{\'Ecole Polytechnique, CMAP, 91128, Palaiseau, France, joffrey.derchu@polytechnique.edu. This work benefits from the financial support of the ERC Grant 679836 Staqamof, the Chaires Analytics and Models for Regulation, Financial Risk, Deep finance and statistics, Machine learning and systematic methods. The author would like to thank Thibaut Mastrolia and Mathieu Rosenbaum for carefully reading this paper.}}

\begin{document}

\maketitle
\begin{abstract}
We introduce a simple framework in which market participants update their prior about an efficient price with a model-based learning process. We show that exponential intensities for the arrival of aggressive orders arise naturally in this setting. Our approach allows us to fully describe market dynamics in the case with Brownian efficient price and informed market takers. We are also able to revisit the emergence of market impact due to meta-order splitting, making several connections with existing literature.

\end{abstract}

\section{Introduction}

The way information about an asset is conveyed to market participants and generates price dynamics is usually referred to as the price formation process. Its study is of both theoretical and practical interest. In particular it explains how economic trends are reflected in the market and is the key to understand and optimize trading costs due to market impact when designing trading algorithms, see \cite{almgc,lo}.
In this paper, we present a framework based on the simple idea that market participants have a model for an unobserved efficient price which should describe the fair value of the asset. They believe that market dynamics, by equating offer and demand, provide information on this price, and they try to estimate this efficient price by continuously updating their view about it using Bayesian updates. The goal of this paper is to derive the dynamics of the market participants' views on this price and of the mid-price in this framework. By using a filtering approach in continuous time to model the learning of the efficient price by the market participants, we are thus able to describe a general mechanism of price formation.\\

We consider a market with a single asset whose efficient price is non observable. This is classical in the study of price and liquidity dynamics, see for example \cite{delattre,gueantlehalle,jaisson,madhavan, stoikov}. In \cite{jaisson} for example, a ``fair price'' is used to link the spread and the impact of market orders by making a zero Profits and Loss assumption for the market makers. In our model there are three types of market participants:
\begin{itemize}
    \item Market observers: these market participants have a view on the efficient price and update it continuously according to the observed market dynamics. They can either stand for high-frequency traders who compute metrics in order to optimize some algorithms, or for softwares which provide a list of indicators for trading desks. In our framework they simply try to estimate the efficient price and they do not interact with the market, except for one particular observer, called the market maker. The market maker is a special observer, as he is responsible for setting the bid and ask prices, and may use his personal metrics to do so. There is only one market maker. We consider two different settings: either the bid and ask prices are fixed (independently of everything else), or they depend on the view of the market maker on the efficient price.
    \item Market takers: these market participants know the efficient price in our model, and trigger trades. In the following we refer to a trade at the ask price as a buy trade and a trade at the bid price as a sell trade. We make the assumption that those trades can be described by two jump processes (one for buy trades and one for sell trades) with intensities which are functions of the difference between the efficient price and the bid or ask price, so the market takers are somewhat opportunistic. In that sense our approach is close to that of \cite{delattre}.
    \item Meta-traders: these market participants trigger large sequences of trades (called meta-orders) according to some schedule which does not depend on the efficient price. We will only consider such participants in Section \ref{sec::meta}.
\end{itemize}
The market observers see the bid and ask prices continuously as well as the trades. For the sake of tractability we assume that they consider that all the trades are initiated by the opportunistic market takers. As a consequence the meta-orders will be indistinguishable from informed trades in Section \ref{sec::meta} and thus they will be treated as a perturbation. It is similar to the approach of \cite{jusselin}. We construct from the observed trades the learning process of a given observer, and we derive how each trade modifies his view on the efficient price. In particular we investigate how the trades lead to market impact. Time is continuous in our model, contrary to the Kyle model \cite{Kyle1985ContinuousAA} or the model of \cite{lillo} on market impact. This allows us in particular to describe how the views evolve in the intervals between the trades and how the bid and ask prices shape the views on the efficient price. Moreover our model uses neither a non-arbitrage argument as in \cite{lillo} or \cite{jaisson}, nor the fact that some agent optimizes a utility function.\\

Here, similarly to \cite{delattre}, we suppose that the market observers believe that the trades contain information on the efficient price: the smaller the difference between the efficient price and the bid or ask price, the more frequent the trades, but the trades' arrival remains random. Instead of estimating the functional linking the efficient price to the intensity of market orders and then recovering the efficient price by making the time horizon go to infinity as in \cite{delattre}, we assume that the observers know this functional and they try to estimate the efficient price online. Their knowledge about the efficient price takes the form of a probability distribution\footnote{We might also refer to this distribution by the names ``prior'' or ``posterior'' depending on the context.}, which represents the probability for the current efficient price to lie in some interval, given the initial knowledge and the information gathered by observing the market orders. Practical motivations for trying to estimate an efficient price include optimizing decision making by better predicting order flow, getting a view of what other participants deem a reasonable price, or anticipate the reaction of the market to the order flow. Also, the views on an efficient price reflect in some way the general consensus on the right price to trade at, so studying the learning process gives us insights into the way prices are formed.\\

In this framework, the shape of the function describing the arrival intensity of market orders is important. A classical hypothesis is to take a decreasing exponential function of the one-sided spread $d$\footnote{By this we mean the distance between the trade price (either the bid price or the ask price) and the efficient price.} of the form
$$\lambda(d)=\lambda_0e^{-ad}$$ for the intensity of arrival of market orders\footnote{Power-law distributions have also been proposed, see for example \cite{stat_prop_bouch}.}. This functional form has been particularly interesting to obtain explicit results in the optimal execution literature, see \cite{avel_stoik,cartea2015algorithmic,gueantlehalle} and in regulation, see \cite{laruelle_mifid2}. Our framework enables us to show that exponential functions are the only natural candidates for the intensity of arrival of market orders, as they appear to be the only possible functions which satisfy two particular properties for our learning process. On the one hand, they ensure that, information on the efficient price is brought to the market mostly through the timestamps of the transactions and their clustering patterns. On the other hand, the exponential intensity of orders decouples the mid-price and the spread in the dynamics of the learning process. While the mid-price becomes the price towards which our estimation of the efficient price tends to go, the spread only dictates how fast the observer accepts the mid-price.\\

Then by setting $\lambda$ as a decreasing exponential function, we are able to derive closed form formulas to describe the price formation process. Taking the particular example of a fixed efficient price, we show that between two trades, and due to the symmetry between buying and selling, an observer's estimation of the efficient price tends to the mid-price, with a characteristic time
$$t_1=\frac{e^{a\frac{S^a-S^b}{2}}}{2\lambda_0},$$
where $a$ and $\lambda_0$ are the parameters used previously in the definition of the intensity function, and $S^b$ and $S^a$ are the current bid and ask prices.
Although fixed bid and ask prices are unrealistic, this last result gives some insights into illiquid markets. In particular, to justify this approach, we give a stability result for the market between two trades: if the bid and ask prices are set symmetrically around the average of a symmetric prior, then the mid-price will not change between two trades.\\

We then consider an approximation of our learning process with dynamic bid and ask prices. In this approximation, if the initial prior is Gaussian, then the posterior remains Gaussian. Furthermore, we are able to derive explicit expressions for the mean and the variance of the posterior if the efficient price is constant or follows a Brownian motion with known volatility. The variance converges to a constant $\sigma_\infty^2$ given by
$$\sigma_\infty^2=\frac{\sigma\sqrt{t_1}}{a}$$
where $\sigma$ is the volatility of the efficient price. This variance can be understood as an observer's asymptotic confidence in his estimation of the efficient price. Our model also allows us to derive formulas for the market impact of a meta-order. While \cite{Kyle1985ContinuousAA} predicts an impact which increases linearly in the total traded volume, it has been largely disproved empirically for small time horizons, see for example \cite{Almgren2005DirectEO,bouchbona,wael,torre}. The square root law has been widely accepted to describe the impact as a function of the traded volume, see \cite{bershova,torre,sq_opt} for instance. However, other exponents, see \cite{Almgren2005DirectEO,ferraris,kissel,moro}, or even logarithmic functions, see \cite{Bouchaud2008HowMS,zarinelli}, have also been considered. Our model, though it is certainly stylized, helps better understand how certain impact shapes can emerge, and why this problem can be controversial. The main theoretical explanations for the shape of the market impact connect it to the persistence of the order flow, see \cite{bouchgefen,jaissonant,jusselin,longmem}, or the size distribution of meta-orders, see \cite{lillo}. For example, in \cite{lillo}, the authors link the size distribution of market orders to the impact by assuming that the market observers try to guess whether a meta-order is being executed and that the market makers make no overall profit. In \cite{jaissonant}, the author supposes that permanent impact is linear in the size of the meta-order and links the long memory exponent of the sign of market orders to the transient market impact.\\

Here we do not assume any distribution for the size of meta-orders nor any intricate strategy by the market makers. We simply look at how the estimation of the efficient price moves when trades happen and thus market impact is driven by the hidden estimates of the efficient price each market observer can have. The long memory of market dynamics is somehow replaced here by the modeling assumption that market orders follow a jump process with an intensity which depends on the efficient price. This is in a way similar to the approach of \cite{toth} where the shape of a hidden latent order book is responsible for market impact form. We will consider two cases: in one case we fix the bid and ask prices and in a second case we make the market maker an observer who has some rule to set the bid and ask prices given his own view on the efficient price. In the first case we define market impact as the change in the prior of an observer, and in the second case we define it as the change in the mid-price (which is directly linked to the change in the prior of the market maker). In both cases we find concave market impacts. For example, in the second case and if we assume that the spreads are small, then the expected impact on the mid-price of a meta-order executed between $t=0$ and $t=T$ with $\beta$ buy orders per second is
\begin{align*}
    \frac{\beta t_1}{a}e^{-\frac{a\sigma}{\sqrt{t_1}}t}(e^{\frac{a\sigma}{\sqrt{t_1}}(t\wedge T)}-1).
\end{align*}
This is linear in the trading speed and concave in time. It is also bounded when $T$ goes to infinity, which can be explained by the fact that the market takers still trade according to the difference between the efficient price, which they know, and the mid-price, which they observe. In this framework the efficient price is not impacted so at some point the meta-order is balanced out by the orders from the opportunistic market takers. We also show that under the same assumptions the impact becomes linear if no information about the efficient price is received while the meta-order is executed \textit{i.e.} if no trades are triggered by opportunistic market takers, and we link the slope of the impact to the asymptotic confidence $\sigma_\infty^2$.\\

For general spreads and for a fixed efficient price, we consider two limiting regimes. For a fast meta-order, the market impact is logarithmic in the traded volume, and for a slow meta-order it is approximately constant and equal to $\frac{\beta t_1}{a}$. This recovers the infinite slope at very short time, as in the case of the square-root law, see \cite{Bouchaud2008HowMS,toth}. For intermediate speeds we give a recursive formula where we show that the $\sinh$ of the impact is the key quantity to compute. Interestingly the market impact in our model is linked directly to the intensity function of the market orders: in most cases we can expect an $\textnormal{arcsinh}$ impact, which may appear similar to a square-root impact. Linear impact and logarithmic impact can appear in limiting cases.\\

Our paper considers only the case where the information on the efficient price is contained in the intensity of the arrival of the aggressive orders. As a consequence we only derive results regarding transient market impact, and the model does not account for market movements which are due to an optimization of quotes in response to a meta-order and a fixed prior, \textit{i.e.} adverse selection by the market maker related to the detection of the meta-order. The study of a model in which the market takers are themselves learning and thus create nontrivial market dynamics is left for further research.\\

The paper is organized as follows. We start by introducing the framework and the modelling choices in Section \ref{sec::model}. Then we consider the case of a fixed efficient price in Section \ref{sec::fixeff}. In Section \ref{sec::approx} we introduce an approximation which yields closed form dynamics. We study the impact of meta-orders in Section \ref{sec::meta}. Finally, some proofs are relegated to the appendix.

\section{Model description}\label{sec::model}
We consider a market participant who tries to estimate the efficient price of the asset $S$ which follows a one-dimensional process. 
The market participant is an observer, and he does Bayesian updates on his prior for $S$ (a probability distribution on $\R$) given a model on the dynamics of $S$ and the available information. We refer to \cite{lipcer2001statistics} for an introduction to Bayesian filtering, \cite{cvitanic2006filtering} or \cite{frey2001nonlinear} for financial applications and \cite{RUPNIKPOKLUKAR2006558} for an extension to jump processes. In our setting, the available information is made of the history of the bid and ask prices and of the trades.

\subsection{Framework}\label{subsec::frame}
We consider a filtered probability space $(\Omega,\F,\P)$ on which the dynamic of the efficient price $S$ is given by
\begin{align}\label{eq::model}
    dS_t = \mu_S(t,S_t)dt + \sigma_S(t,S_t)dW_t
\end{align}
where $W$ is a Brownian motion, $\mu_S$ and $\sigma_S$ are two real-valued continuous Lipschitz functions, which are known by the observer. In particular there is strong existence and uniqueness of $S$ given some initial value $S_0\in\R$. The initial efficient price $S_0$ follows some probability distribution $\pi^0$ on $\R$.
For any $t\geq 0$, define the operator $\Lc_t$ on $\Cc^2_c(\R)$ the set of twice differentiable functions with continuous first and second derivatives and with compact support by 
\begin{align*}
    \Lc_t f = \mu_S(t,.)\partial_x f + \frac{1}{2}\sigma_S(t,.)^2\partial^2_{xx} f
\end{align*}
for any $f\in \Cc^2_c(\R)$. Also define its adjoint operator $\Lc_t^*$ by $\Lc_t^* f = -\partial_x(\mu_S(t,.) f) + \partial^2_{xx}(\frac{1}{2}\sigma_S(t,.)^2 f)$.\\

Let the ask and bid prices $S^a$ and $S^b$ be two càdlàg processes to be specified later. Let $N^a$ and $N^b$ be two jump processes with unit jumps and compensators given by $\lambda(S^a_{t-}-S_t)$ and $\lambda(S_t-S^b_{t-})$ on the ask and bid side respectively. The intensity function $\lambda$ is a continuous, non-negative and decreasing function on $\R$ (see Section \ref{sec::exp_int} for more details on the choice of $\lambda$). The processes $N^a$ and $N^b$ are respectively the number of trades on the ask side and the bid side. Trades happen on the ask side with an intensity $\lambda(S^a_{t-}-S_t)$ which is decreasing with respect to the distance between the ask price and the efficient price. Similarly trades happen on the bid side with an intensity $\lambda(S_t-S^b_{t-})$ which is decreasing in the distance between the efficient price and the bid price. Let $(\Fc_t)_{t\geq 0}$ be the filtration associated with an observer, defined as the completion of the filtration generated by $S^a$, $S^b$, $N^a$ and $N^b$. Note that neither $W$ nor $S$ are observable. We denote by $\E$ the expectation under $\P$.

\begin{remark}
We do not consider the case where the observer uses erroneous dynamics in his estimation of the efficient price. However the filtration is generated only by $S^a$, $S^b$, $N^a$ and $N^b$, and most of our results aim at expressing the learning process using only those four processes (and not the law of $S$). So, even if the participants had erroneous models, most of our results could still be used to understand the market dynamics.
\end{remark}

\subsection{Filtering equation}
We consider an observer who tries to estimate $S$ from the filtration $(\Fc_t)_{t\geq 0}$. He adopts a purely Bayesian point of view, \textit{i.e.} he aims at computing $\E[f(S_t)|\Fc_t]$ for any bonded function $f$. At time $t=0$, he has a prior on $S$ given by the probability measure $\pi^0$ on $\R $. Then he updates this measure given the information he receives. We now recall the filtering equations associated to the filtration $(\Fc_t)_{t\geq 0}$, see for example \cite{RUPNIKPOKLUKAR2006558}.\\

We say that a process $\rho$ on the space of measures on $\R $ defined for any $f\in \Cc_c^2(\R  )$ by $\rho[f]_t=\int f(s)d\rho_t(s)$ is a solution of the Zakai equation if $\rho[f]_t$ is càdlàg and if
\begin{equation}\label{eq::zakai}
    \begin{split}
    d\rho[f]_t = & \rho[\Lc_t f]_t dt +  (\rho[f(.)\lambda(S_{t-}^a-.)]_{t-} - \rho[f]_{t-}) (dN^a_t-dt)\\ &+ (\rho[f(.)\lambda(.-S_{t-}^b)]_{t-} - \rho[f]_{t-}) (dN^b_t-dt).
    \end{split}
\end{equation}
for any $f\in \Cc_c^2(\R  )$. \\

Classical uniqueness results ensure that the posterior is fully determined by the solution of this equation if $S^a$ and $S^b$ are deterministic. We indeed have the following result.
\begin{proposition}\label{lem::uniq}
Suppose that $S^a$ and $S^b$ are deterministic and that 
$$0<\lambda_-\leq\lambda\leq\lambda_+$$
for some constants $\lambda_-$, $\lambda_+$, and $\lambda$ the intensity function defined in \eqref{subsec::frame}.
Then, for any $t\geq 0$ and $f\in \Cc_c^2(\R)$, $\rho[f]_t=\bar\E[f(S_t)|\Fc_t]$ is the unique solution to the Zakai equation \eqref{eq::zakai} with $\rho_0=\pi^0$. The expectation $\bar\E$ is taken under the measure $\bar\P$ given by the change of probability $\frac{d\bar\P}{d\P}|_t=\Gamma_t$ with
\begin{align*}
    \Gamma_t = e^{\int_0^t-\log(\lambda(S^a_{s-}-S_s))dN^a_s-\int_0^t(1-\lambda(S^a_{s-}-S_s))ds}e^{\int_0^t-\log(\lambda(S_s-S^b_{s-}))dN^b_s-\int_0^t(1-\lambda(S_s-S^b_{s-}))ds}.
\end{align*}
The law of $S_t$ given $\Fc_t$ is then described by $\pi[f]_t = \E[f(S_t)|\Fc_t]=\frac{\rho[f]_t}{\rho[1]_t}$.
\end{proposition}
\begin{proof}
The proof can be found in \cite{qiao2}, see also \cite{ceci,CLAUDI-2012,qiao,mandrekar}.
\end{proof}
The measure $\rho$ is called the unnormalized filter, as opposed to $\pi$, which we call the normalized filter. If the solution of the Zakai equation has a smooth density $\hat m$, an integration by parts shows that it is a solution to the following Zakai SPDE
\begin{equation}\label{eq::zakai_spde}
    \begin{split}
    d\hat m_t(x) = & (\Lc^*_t \hat m_t)(x)dt\\
    &+ \hat m_{t-}(x)(\lambda(S_{t-}^a-x) - 1) (dN^a_t-dt) \\
    &+ \hat m_{t-}(x)(\lambda(x-S_{t-}^b) - 1) (dN^b_t-dt).    
    \end{split}
\end{equation}

\begin{remark}
Similarly, under the same assumptions as in Proposition \ref{lem::uniq}, the normalized filter $\pi$ is the unique solution of the Kushner-Stratonovich equation: for any $f\in \Cc_c^2(\R  )$, $\pi[f]_t$ is càdlàg and
\begin{equation}\label{eq::ks}
    \begin{split}
    d\pi[f]_t = &\pi[\Lc_t f]_t dt +(-\pi[f(.)(\lambda(S_t^a-.)+\lambda(.-S_t^b))]_t +\pi[f]_t\pi[(\lambda(S_t^a-.)+\lambda(.-S_t^b))]_t)dt \\
    &+ \big(\frac{\pi[f(.)\lambda(S_{t-}^a-.)]_{t-}}{\pi[\lambda(S_{t-}^a-.)]_{t-}} - \pi[f]_{t-}\big) dN^a_t \\
    &+ \big(\frac{\pi[f(.)\lambda(.-S_{t-}^b)]_{t-}}{\pi[\lambda(.-S_{t-}^b)]_{t-}} - \pi[f]_{t-}\big) dN^b_t.
    \end{split}
\end{equation}
for any $f\in \Cc_c^2(\R  )$.
\end{remark}
\comment{Similarly, the normalized density $m$ is the solution to the Kushner-Stratonovich SPDE
\begin{equation}\label{eq::ks_spde}
    \begin{split}
    dm_t(x) = & (\Lc^*_t m)(t,x)dt\\
    &-m_t(x)(\int (\Lc^*_t m)(t,y)dy]dt \\
    &-m_t(x)(\lambda(S_t^a-x)+\lambda(x-S_t^b)))dt\\
    &+m_t(x)(\int (\lambda(S_t^a-y)+\lambda(y-S_t^b))m_t(y)dy)dt \\
    &+ m_{t-}(x)\big(\frac{\lambda(S_{t-}^a-x)}{\int\lambda(S_{t-}^a-y)m_{t-}(y)dy} - 1\big) dN^a_t \\
    &+ m_{t-}(x)\big(\frac{\lambda(x-S_{t-}^b)}{\int\lambda(y-S_{t-}^b)m_{t-}(y)dy} - 1\big) dN^b_t.    
    \end{split}
\end{equation}}
Most of the time we work with unnormalized densities as the Zakai SPDE \eqref{eq::zakai_spde} is linear. Also, if the initial prior has a density, a solution $\hat m$ of the Zakai SPDE \eqref{eq::zakai_spde} (if it exists) gives the density of the unique unnormalized filter. This is only a technical tool as we can retrieve the density $m$ of the probability distribution from the density $\hat m$ of the normalized filter by renormalizing $\hat m$: $m_t(x) = \frac{\hat m_t(x)}{\int \hat m_t(y)dy }$.\\

There are three distinct terms in \eqref{eq::zakai_spde}:
\begin{itemize}
    \item The first term $\Lc^*_t \hat m_t = -\mu_S \partial_x \hat m_t +\frac{1}{2}\sigma_S^2 \partial_{xx}^2\hat m_t$ takes into account the model of the observer: the density will diffuse and drift according to the model.
    \item The second term $\hat m_{t-}(x)(\lambda(S_{t-}^a-x) - 1) dN^a_t+\hat m_{t-}(x)(\lambda(x-S_{t-}^b) - 1) dN^b_t$ describes what happens after each trade. Each time a trade occurs at the ask price the density is multiplied by $\lambda(S_{t-}^a-x)$ so the new density puts more weight to the higher prices. Conversely each time a trade occurs at the bid price the density is multiplied by $\lambda(x-S_{t-}^b)$ so the new density puts more weight to the lower prices. 
    \item The third term $-\hat m_t(x)(\lambda(S_t^a-x) +\lambda(x-S_t^b) - 2) dt$ explains the behaviour of the density between two trades. The density is modified by some potential $x\mapsto\lambda(S_t^a-x) +\lambda(x-S_t^b)$ which has its extremum at $\frac{S^a_t+S^b_t}{2}$. Note that the $2\hat m_t(x)dt$ term is irrelevant from the point of view of a learning observer as adding a term of the form $\hat m_t(x)\phi_t$ in the Zakai SPDE \eqref{eq::zakai_spde}, with $\phi$ some c\`adl\`ag process which does not depend on $x$, modifies $\hat m_t(x)$ but not $\hat m_t(x)/\int \hat m_t(y)dy$.
\end{itemize} 
\comment{The Kushner-Stratonovich SPDE \eqref{eq::ks_spde} is essentially the same. The non-local terms serve only to ensure that the solution has the same total weight as the initial density $m_0$.
\begin{itemize}
    \item The $-m_t(x)(\int (\Lc^* m)(t,y)dy)dt+m_t(x)(\int (\lambda(S_t^a-y)+\lambda(y-S_t^b))m_t(y)dy)dt$ bit adds a multiplicative term which does not depend on $x$.
    \item When a trade happens, the density is renormalized after the update. If a trade happens at the ask price the density is multiplied by $\frac{\lambda(S_{t-}^a-x)}{\int\lambda(S_{t-}^a-y)m_{t-}(y)dy}$ instead of $\lambda(S_{t-}^a-x)$, and if a trade happens at the bid price the density is multiplied by $\frac{\lambda(x-S_{t-}^b)}{\int\lambda(y-S_{t-}^b)m_{t-}(y)dy}$ instead of $\lambda(x-S_{t-}^b)$. 
\end{itemize}
}
\begin{remark}
The boundedness assumption on $\lambda$ in Proposition \ref{lem::uniq} is only technical. If we take $\lambda$ to be a positive, continuous, decreasing and convex function, then clipping $\lambda$ makes the assumption trivial, and it has no impact on the financial interpretations or the numerical results if the clipping is done far enough. So most of the time we will only consider the Zakai SPDE, ignoring the boundedness condition.
\end{remark}

\begin{remark}
The Zakai equation describes the evolution of the posterior of each observer, if they suppose that the price process follows the dynamics given by \eqref{eq::model}. Each observer could have different functions $\mu_S$ and $\sigma_S$. It is possible to complexify the dynamics by making our observers learn some parameters on which $\mu_S$ and $\sigma_S$ could depend. Thus they could revise their views on the dynamics themselves.
\end{remark}

In the following, we will consider two possible cases:
\begin{itemize}
    \item $S^a$ and $S^b$ are fixed constants so an observer uses the Zakai filtering equation \eqref{eq::zakai_spde} to update his views on the efficient price. This toy model will help us understand some important features, and it will give insights into illiquid markets.
    \item Our observer is the market maker. He applies the same filtering equation \eqref{eq::zakai_spde} to update his prior through the observation of $N^a$ and $N^b$. He then uses his posterior to change his quotes. This filter, together with a rule for the update of the quotes given the posterior and the definition of $N^a$ and $N^b$ yield a fixed-point problem for $S^a$ and $S^b$ which we discuss in Proposition \ref{prop::stabil} and in Sections \ref{sec::approx} and \ref{sec::meta}.
\end{itemize}

\subsection{Exponential intensities as a consequence of microstructure}\label{sec::exp_int}

The intensity function $\lambda$ plays an important role. For financial reasons we consider continuous, strictly decreasing and convex functions. We introduce two new properties which arise naturally in our framework and we show that they are satisfied only by exponential functions.\\

The first property translates the idea that, given the pre-trade prior density $m_{t-}$, the jump in the filter happening at time $t$ because of a trade should not depend on the bid and ask prices $S^a_{t-}$ and $S^b_{t-}$. In more financial terms, it means that knowing only the value of a trading price is not sufficient to estimate the efficient price: an observer needs the timestamps and the clustering patterns of the trades. It can be translated into the following mathematical property.
\paragraph{Property (a)} $\lambda$ is a continuous, positive, strictly decreasing, exponentially bounded\footnote{By this we mean that $\lambda(x)<e^{c|x|}$ for some $c>0$.} and convex function. Also, the maps 
\begin{align*}
    \Mc^a_\lambda&\rightarrow \R^{\R^2}\\
    m&\mapsto\big((z,x)\mapsto \frac{\lambda(z-x)m(x)}{\int \lambda(z-y)m(y) dy}\big)
\end{align*} and 
\begin{align*}
    \Mc^b_\lambda&\rightarrow \R^{\R^2}\\
    m&\mapsto\big((z,x)\mapsto \frac{\lambda(x-z)m(x)}{\int \lambda(y-z)m(y) dy}\big)
\end{align*}
where $\Mc^a_\lambda = \{m> 0, \int m(x)=  1, 0<\int \lambda(y-x)m(x)dx<\infty, \forall y\in\R \}$ and $\Mc^b_\lambda = \{m> 0, \int m(x)=  1, 0<\int \lambda(x-y)m(x)dx<\infty, \forall y\in\R \}$, have their image in the set of functions which do not depend on $z$.\\

Note that $\Mc^a_\lambda$ and $\Mc^b_\lambda$ are always non-empty for $\lambda$ a continuous, strictly decreasing, exponentially bounded and convex function as they contain the Gaussian density functions.\\

The two maps in Property (a) describe the jump in the density of the prior when a trade happens, on the ask side or the bid side. Property (a) can be interpreted as follows: $m_t$ depends on $m_{t-}$ and on the sign of the trade but not on $S^a_{t-}$ or $S^b_{t-}$. In other words, for a given observer, a trade on the ask side (resp. on the bid side), by itself, has the same informational value whatever the spread. It does not mean that the bid and ask prices are useless for the observer, but that he needs to observe the dynamics of the trades: for example the time elapsed before the trade is important. In a way learning is intrinsically dynamic: the knowledge of the price of a trade only helps if we know the history of the previous bid and ask prices and of the trades.\\

The second property relates to the fact that in the absence of trades, the mid-price $\frac{S^a+S^b}{2}$ plays the role of the potential new estimator of the efficient price and the half-spread $\frac{S^a-S^b}{2}$ modulates the speed of learning this price. In more financial terms, it means that, if we look at the market at some given time with no specific information, the best available estimate of the efficient price is the mid-price. The spread only helps build confidence in the mid-price. It can be translated into the following mathematical property.
\paragraph{Property (b)} The function $\lambda$ is continuous, positive, strictly decreasing, exponentially bounded, convex and four times differentiable. Also, there exists $3$ functions $f:\R\rightarrow\R$, $g:\R\rightarrow\R$ and $h:\R\rightarrow\R$, such that
\begin{equation*}
    \begin{split}
    -(\lambda(s^a-x) +\lambda(x-s^b)-2) =\underbrace{h(s^a,s^b)}_{\textnormal{global constant}}\underbrace{-g(x-\frac{s^a+s^b}{2})}_{\textnormal{Potential with minimum at }\frac{s^a+s^b}{2}}f(\frac{s^a-s^b}{2})
    \end{split}
\end{equation*}
for all $s^a,s^b,x\in\R$.\\

Property (b) can also be interpreted in more physical terms. It means that the mid-price and the spread play two independent roles for the learning process: the spread appears as a time dilation parameter, while the mid-price is the location of the minimum of a potential. As a consequence, the spread serves only to determine the speed with which the mid-price is accepted as the best estimator of the efficient price when no trade happens. When $\lambda(x)=\lambda_0e^{-ax}$ for some $a>0$ and $\lambda_0> 0$, the term without jump in $d\hat m_t$ is actually
\begin{align*}
     &(\mathcal{L}^*_t\hat m_t)(x)dt-\lambda_0\hat m_t(x)(e^{-a(S^a_t-x)}+e^{-a(x-S^b_t)})dt    =(\mathcal{L}^*_t\hat m_t)(x)dt-2\lambda_0e^{-\frac{S^a_t-S^b_t}{2}}\hat m_t(x)\cosh(\frac{S^a_t+S^b_t}{2}-x)dt,
\end{align*}
so a high spread can be counterbalanced by a high base intensity $\lambda_0$.\\

Our result is the following.
\begin{proposition}\label{prop::lamb}
Property (a) and Property (b) hold if and only if $\lambda$ is of the form $\lambda(x)=\lambda_0e^{-a x}$ for some $\lambda_0>0$ and $a>0$.
\end{proposition}
\begin{proof}
See Appendix \ref{app:proof_prop_lamb}.
\end{proof}

We introduce the following assumption which will be useful in the next sections.
\begin{assumption_lamb*}\label{assump::lamb}
We assume that $\lambda(x)=\lambda_0e^{-ax}$ for some $a>0$ and $\lambda_0> 0$. 
\end{assumption_lamb*}

For $\sigma\geq 0$, let $\Ac_\sigma$ be the following assumption, which will help us get explicit expressions.
\begin{assumption_sig*}\label{assump::sig}
We assume that
\begin{align*}
    \mu_S=0\hspace{2cm}
\end{align*}
and $\sigma_S=\sigma$.
\end{assumption_sig*}

\subsection{Example with a Gaussian prior at time $t-$}\label{subsubsec::ex}
Assume \hyperref[assump::lamb]{$\Ac^\lambda$} and \hyperref[assump::sig]{$\Ac_\sigma$} and suppose that right before time $t$, $m_{t-}$ is Gaussian with mean $\bar S$ and variance $\bar\sigma^2>0$.\footnote{For example this is the case of an observer with initial prior $m_0\sim\Nc(\bar S,\bar\sigma^2)$ who for some reason did not update his views before $t$.} Then
\begin{equation*}
    \begin{split}
    dm_t(x) = &\frac{\sigma^2}{2\bar\sigma^2}(-1+\frac{(x-\bar S)^2}{\bar\sigma^2})m_t(x)dt\\
    &+2\lambda_0 e^{-a\frac{S_t^a-S_t^b}{2}}( -m_t(x)\cosh(a(x-\frac{S^a_t+S^b_t}{2}))+e^{\frac{a^2\bar\sigma^2}{2}}\cosh(a(\bar S-\frac{S^a_t+S^b_t}{2})))dt \\
    &+ m_{t-}(x)\big(\frac{\lambda_0 e^{-a(S_{t-}^a-x)}}{\int\lambda_0 e^{-a(S_{t-}^a-y)}m_{t-}(y)dy} - 1\big) dN^a_t \\
    &+ m_{t-}(x)\big(\frac{\lambda_0 e^{-a(x-S_{t-}^b)}}{\int\lambda_0 e^{-a(y-S_{t-}^b)}m_{t-}(y)dy} - 1\big) dN^b_t.  
    \end{split}
\end{equation*}
A simple computation shows that if $m_{t-}$ is Gaussian right before a jump, the posterior is still Gaussian, with same variance but with a shifted mean $\bar S_{post} = \bar S + a\bar\sigma^2$. Unsurprisingly the estimation of the efficient price is revised upward after an aggressive buy trade. The jump is proportional to $\bar\sigma^2$: the less confident a trader in his estimation, the more he will change it when something happens. It is also proportional to the scale factor $a$ in the intensity function. Indeed, as $a$ grows larger, trades become less frequent, so each aggressive order will be that much more important to estimate the efficient price. Also, from the above equation we deduce that a Gaussian prior does not usually remain Gaussian if it is updated continuously, because of the $\cosh(a(x-\frac{S^a_t+S^b_t}{2})$ term.\\

At time $t$ the mean evolves locally as 
\begin{equation*}
    \begin{split}
    d(\int x m_t(x)dx) = &-2\lambda_0a\bar\sigma^2  e^{-a\frac{S_t^a-S_t^b}{2}+\frac{a^2\bar\sigma^2}{2}} \sinh(a(\bar S-\frac{S^a_t+S^b_t}{2}))dt \\
    &+ a\bar\sigma^2 (dN^a_t -dN^b_t).
    \end{split}
\end{equation*}
As we saw in Section \ref{sec::exp_int} it is revised upward or downward by $a\bar\sigma^2$ if a trade occurs. If no trade happens, the dynamics of the prior depend on the mid-price $\frac{S^a_t+S^b_t}{2}$. We interpret it as a tendency for an observer to believe that the bid and ask prices are set symmetrically around the efficient price if no trade happens, which makes him see the mid-price as a good estimate. For example, if our estimated mean is higher than the mid-price, then we think that there will be more aggressive buys than sells. If nothing happens, we tend to think that the number of buys will be closer to the number of sells, and we decrease our estimate. The spread $S_t^a-S_t^b$ plays also a role: the bigger the spread, the less information we learn by seeing no trade. As before, the value $a\bar\sigma^2$ measures our confidence in our estimates.\\

The mean square evolves locally as
\begin{equation*}
\begin{split}
    d(\int x^2 m_t(x)dx) = &\sigma^2dt-2\lambda_0 e^{-a\frac{S_t^a-S_t^b}{2}+\frac{a^2\bar\sigma^2}{2}}\Big((\bar\sigma^2+\bar S^2+a^2\bar\sigma^4)\cosh\big(a(\bar S-\frac{S^a_t+S^b_t}{2})\big) \\
    &+2a\bar S \sigma^2\sinh\big(a(\bar S-\frac{S^a_t+S^b_t}{2})\big)\Big)dt + 2\hat Sa\bar\sigma^2 (dN^a_t -dN^b_t)  .
\end{split}
\end{equation*}
so our variance evolves locally as 
\begin{equation*}
    d\big(\int x^2 m_t(x)dx-(\int x m_t(x)dx)^2\big) = \Big(\sigma^2-2\lambda_0 e^{-a\frac{S_t^a-S_t^b}{2}+\frac{a^2\bar\sigma^2}{2}}(\bar\sigma^2+\bar S^2+a^2\bar\sigma^4)\cosh\big(a(\bar S-\frac{S^a_t+S^b_t}{2})\big)\Big)dt
\end{equation*}
so it is locally strictly decreasing if it is large and locally strictly increasing if it is small. So we could expect some convergence of the variance (see Section \ref{sec::approx}). Also, the further our estimate is from the mid-price, the slower the variance changes. 

\section{Learning a fixed efficient price}\label{sec::fixeff}
In this section we assume \hyperref[assump::sig]{$\Ac_0$}, \textit{i.e.} we fix the efficient price. We study the learning process in this toy framework, which will give us important intuitions for the more general case. It is also adapted to illiquid markets in which movements in the efficient price are deemed small compared to the uncertainty on it.\\

First we observe that with a constant efficient price the Zakai equation on densities \eqref{eq::zakai_spde} has a simple solution. The proof of the following proposition is a direct application of Itô's formula.

\begin{proposition}\label{prop_simple}
Assume \hyperref[assump::sig]{$\Ac_0$} and suppose that the prior at time $0$ has a density $m_0$. Then the Zakai SPDE \eqref{eq::zakai_spde} has a unique solution $\hat m$ given by
\begin{align*}
    \hat m_t(x) = m_0(x)e^{-\int_0^t(\lambda(S^a_u-x)+\lambda(x-S^b_u)-2)du}e^{\int_0^t\ln{\lambda(S^a_{u-}-x)}dN^a_u+\ln{\lambda(x-S^b_{u-})}dN^b_u}.
\end{align*}
\end{proposition}
We deduce the following corollary.
\begin{corollary}
If additionally \hyperref[assump::sig]{$\Ac^\lambda$} holds and the initial prior is Gaussian, then $\hat{m}_t$ is integrable in $x$ for each $t$.
\end{corollary}

We introduce the following assumption which allows us to derive convergence results of the posterior distribution to the real efficient price.

\begin{assumption_fix*}\label{assump::fix}
We assume that $S^a$ and $S^b$ are fixed constants.
\end{assumption_fix*}
We show later in Proposition \ref{prop::stabil} that this assumption makes sense when we look at what happens between two trades.\\
\comment{
The almost sure convergence of posterior distributions was first studied by \cite{doob} and later by \cite{Schwartz-1964} for example. We give a simple result for our framework, following the idea of \cite{Shalizi_2009}.

\begin{proposition}\label{prop::conv_post}
Assume \hyperref[assump::fix]{$\Ac^\text{fix}$} and \hyperref[assump::sig]{$\Ac_0$}. Suppose that the prior at time $0$ has a density $m_0$ with bounded support. Then the normalized filter has a density $m$ such that
\begin{align*}
    \frac{1}{t}\log(m_t(s))\underset{t\rightarrow+\infty}{\rightarrow}-h(s)
\end{align*}
a.s., for all $s$ such that $m_0(s)\neq 0$, with
\begin{align*}
    h(s) &= (\lambda(S^a-s)+\lambda(s-S^b))-(\lambda(S^a-S)+\lambda(S-S^b))\\
    &+\lambda(S^a-S)\log(\frac{\lambda(S^a-S)}{\lambda(S^a-s)})+\lambda(S-S^b)\log(\frac{\lambda(S-S^b)}{\lambda(s-S^b)}).
\end{align*}
\end{proposition}
\begin{proof}
See Appendix \ref{app:proof_conv_post}.
\end{proof}
As $h$ has a unique minimum at $S$, this result tells us that after a long time we expect an observer to have learned the efficient price if it is fixed and the observer knows it is fixed. Now we look at what happens between two trades.}

We have the following result:
\begin{theorem}\label{prop::dirac}
Assume \hyperref[assump::fix]{$\Ac^\text{fix}$}, \hyperref[assump::sig]{$\Ac_0$}, and that $\lambda$ is strictly decreasing and strictly convex. \comment{Then the probability measures which are steady states of the continuous part of the (measure valued) Kushner-Stratonovich equation are the sums of 2 Dirac distributions at two symmetrical points around $\frac{S^a+S^b}{2}$ (they are possibly the same points).}
If $m_0\in\Cc(\R)$ is the density of a probability measure on $\R$ with $m_0(\frac{S^a+S^b}{2})\neq 0$, then we have the following results.
\begin{enumerate}[(i)]
    \item Let $\tau=\inf\{t>0,\max(N^a_t,N^b_t)>0\}$ be the time of the first trade. There exists a.s. a unique solution $\hat m$ on $[0,\tau)$ to the Zakai SPDE with initial density $m_0$, and this solution is integrable. The renormalized probability density $u_t=\frac{\hat m_t}{\int \hat m_t(x)dx}$ is deterministic, and it tends in the sense of measures to a Dirac distribution at point $\frac{S^a+S^b}{2}$ as $t\rightarrow+\infty$. 
    \item In particular if \hyperref[assump::lamb]{$\Ac^\lambda$} holds, then
    \begin{align*}
        u_t(x)\underset{t\rightarrow+\infty}{\sim} \sqrt{\frac{t}{\pi t_1}}\frac{m_0(x)}{m_0(\frac{S^a+S^b}{2})}e^{-\frac{t}{t_1}(\cosh(x-\frac{S^a+S^b}{2})-1)},
    \end{align*}
    where $t_1 = \frac{e^{a\frac{S^a-S^b}{2}}}{2\lambda_0}$.

\end{enumerate}
\end{theorem}
\begin{proof}
See Appendix \ref{app::proof_dirac}.
\end{proof}
\comment{The first part of this theorem shows that the only possible stationary solutions are symmetrical with respect to the mid-price. Indeed, as the bid and ask prices are symmetrical in our model, the absence of trades means that the bid and ask prices are each as far from the efficient price as the other.
The second part of this theorem adds an assumption on the initial prior.} This theorem implies that if the initial prior has a density whose support contains the mid-price, then necessarily an observer learns this mid-price during the intervals between the trades. This learning takes places with a characteristic time $t_1$. Note that $t_1$ is decreasing in the spread.\\

We now look at the case where the market maker learns the efficient price: we impose ex post that the mean of $m_t$ is the mid-price and we look for a fixed point. We prove a stability result: the market maker does not change his quotes if no trade occurs.

\begin{proposition}\label{prop::stabil}
Assume \hyperref[assump::lamb]{$\Ac^\lambda$} and \hyperref[assump::sig]{$\Ac_0$}, that $m_0$ is Gaussian, and suppose that the processes $S^a$ and $S^b$ satisfy $\frac{S^a-S^b}{2}=\delta>0$ and are such that the solution of the Zakai SPDE is a distribution with mean $\frac{S^a+S^b}{2}$ (\textit{i.e.} that the market maker chooses his mid-price as the mean of $m_t$, and that he takes a constant half-spread $\delta$). Then the quotes are stable, \textit{i.e.} $\frac{S^a_t+S^b_t}{2}=\frac{S^a_0+S^b_0}{2}$, for $t<\tau=\inf\{t>0,\max(N^a_t,N^b_t)>0\}$. Also, there exists a.s. a unique solution $\hat m$ on $[0,\tau)$ to the Zakai SPDE with initial density $m_0$, and this solution is integrable. The renormalized probability density $u_t=\frac{\hat m_t}{\int \hat m_t(x)dx}$ is deterministic, and it tends in the sense of measures to a Dirac distribution at point $\frac{S^a_0+S^b_0}{2}$ as $t\rightarrow+\infty$.
\end{proposition}
\begin{proof}
Note from Proposition \ref{prop_simple} that the unique solution $\hat m$ of the Zakai SPDE has an integrable continuous density for $t<\tau$, with moments of every order, so the renormalized distribution $m$ solves the Kushner-Stratonovich SPDE for $t<\tau$. Writing $x_t=\int yu_t(y)dy$, observe from \eqref{eq::zakai_spde} that
\begin{align*}
    &\hat{m}_t(x_t-x)= m_0(x_t-x)e^{2t-\frac{t}{t_1} \cosh(ax)}
\end{align*}
for $t<\tau$. Integrating the Kushner-Stratonovich SPDE, we see that
\begin{align*}
    dx_t &= -\frac{dt}{t_1}\int(x_t-x)u_t(x_t-x)(\cosh(ax)-\int\cosh(ay)u_t(x_t-y)dy)dx.
\end{align*}
Plugging 
\begin{align*}
    u_t(x_t-x)= \frac{m_0(x_t-x)e^{2t-\frac{t}{t_1} \cosh(ax)}}{\int m_0(x_t-y)e^{2t-\frac{t}{t_1} \cosh(ay)}dy}
\end{align*}
into the dynamics of $x_t$ we find that $x_t$ solves a continuous and locally Lipschitz ODE. As a consequence $x_t = x_0$ \textit{i.e.} $\frac{S^a_t+S^b_t}{2}=\frac{S^a_0+S^b_0}{2}$ is the only possible solution, for $t<\tau$. Plugging $x_t=x_0$ in the previous equation we obtain the convergence of $u_t$, conditional on $t<\tau$ to a Dirac distribution at $x_0$.
\end{proof}
This result completes Theorem \ref{prop::dirac} by giving the point of view of the market maker between trades. It shows that a market maker who quotes a constant spread has no incentive to move his bid and ask prices if his initial prior on the efficient price is Gaussian.

\section{Small spread approximation}\label{sec::approx}
Assume \hyperref[assump::lamb]{$\Ac^\lambda$} and \hyperref[assump::sig]{$\Ac_\sigma$} for some $\sigma\geq 0$. Let $m_0$ be the density of a $\Nc(x_0,\sigma_0^2)$, and suppose that $\frac{S_t^a-S_t^b}{2}=\delta\geq 0$ for all $t$. The Zakai SPDE is
\begin{equation*}
    \begin{split}
    du_t(x) = &\frac{\sigma^2}{2}\partial^2_{xx}u_t(x)dt-\frac{1}{t_1} \cosh(a(x-\frac{S^a_t+S^b_t}{2}))u_t(x)dt \\
    &+ u_{t-}(x)(\lambda_0 e^{-a(S_t^a-x)} - 1) dN^a_t + u_{t-}(x)(\lambda_0 e^{-a(x-S_t^b)} - 1) dN^b_t.
    \end{split}
\end{equation*}
We consider an approximation of this SPDE which we obtain by replacing the $\cosh$ term by the first two terms in its Taylor expansion:

\begin{equation}\label{eq::aprox_zakai}
    \begin{split}
    du_t(x) = &\frac{\sigma^2}{2}\partial^2_{xx}u_t(x)dt-\frac{1}{t_1}u_t(x)dt-\frac{1}{2t_1} a^2(x-\frac{S^a_t+S^b_t}{2})^2u_t(x)dt \\
    &+ u_{t-}(x)(\lambda_0 e^{-a(S_t^a-x)} - 1) dN^a_t + u_{t-}(x)(\lambda_0 e^{-a(x-S_t^b)} - 1) dN^b_t.  
    \end{split}
\end{equation}

This new equation approximates the learning process in the case where typical spreads are small compared to $1/a$ where $a$ is the scale parameter given by the intensity function in \hyperref[assump::lamb]{$\Ac^\lambda$}. Our goal here is not to make this approximation rigorous, but to use it to obtain simple explicit and reasonable results. Also, note that a solution to \eqref{eq::aprox_zakai} is not necessarily the density of a probability distribution.

\subsection{Fixed price}
Throughout this subsection we assume that \hyperref[assump::sig]{$\Ac_0$} is satisfied.
The proof of the following proposition is a direct application of Itô's formula.
\begin{proposition}\label{prop::fix_app}
There exists a unique solution to \eqref{eq::aprox_zakai} denoted by $u$ such that $u_0=m_0$ and given for any $t\geq 0$ by
\begin{align*}
    u_t(x) &= m_0(x)e^{-\frac{1}{t_1}\frac{1}{2}a^2\int_0^t(\frac{S^a_s+S^b_s}{2}-x)^2ds+ax(N^a_t-N^b_t)-\frac{t}{t_1}}.
\end{align*}
In particular, $u_t$ is (up to a multiplicative constant) the density of a Gaussian distribution $\Nc(x_t,\sigma_t^2)$ where
\begin{align*}
    x_t = \frac{\frac{x_0}{\sigma_0^2}+a^2\frac{1}{t_1}\int_0^t\frac{S^a_s+S^b_s}{2}ds+a(N^a_t-N^b_t)}{\frac{1}{\sigma_0^2}+a^2\frac{t}{t_1}}
\end{align*}
and
\begin{align*}
    \sigma_t^2 = \frac{1}{\frac{1}{\sigma_0^2}+a^2\frac{t}{t_1}}\underset{t\rightarrow+\infty}{\rightarrow}0.
\end{align*}
\end{proposition}
\begin{remark}
This solution should be compared to the exact solution
\begin{align*}
    \hat m_t(x) &= m_0(x)e^{-\frac{1}{t_1}\int_0^t\cosh(a(\frac{S^a_s+S^b_s}{2}-x))ds+ax(N^a_t-N^b_t)}.
\end{align*}
One advantage of the approximation is that it helps us derive explicit formulas for the moments and the dynamics of the posterior.
\end{remark}
We deduce the following results by specifying $S^a$ and $S^b$.
\begin{corollary}
Under \hyperref[assump::fix]{$\Ac^\text{fix}$} we have
\begin{align*}
    x_t\underset{t\rightarrow+\infty}{\rightarrow}\frac{a^2\frac{1}{t_1}\frac{S^a+S^b}{2}+a\frac{1}{t_1}\sinh(-a(\frac{S^a+S^b}{2}-S_0))}{a^2\frac{1}{t_1}},
\end{align*}
a.s.. If instead we impose $N^a=N^b=0$, then
\begin{align*}
    &x_t\underset{t\rightarrow+\infty}{\rightarrow}\frac{S^a+S^b}{2}.
\end{align*}
If $\frac{S^a_t+S^b_t}{2}=x_t$ with fixed half-spread $\frac{S^a_t-S^b_t}{2}=\delta>0$\footnote{This means that the particular observer which learning process we are considering is the market maker, and that he chooses the mid-price as $x_t$ with a constant half-spread $\delta$.}, then $x_t$ is the unique solution to
\begin{align*}
    x_t=x_0+\int_0^t \frac{a}{\frac{1}{\sigma_0^2}+a^2\frac{s}{t_1}}\big(dN^a_s-dN^b_s\big).
\end{align*}
\end{corollary}
\begin{proof}
The proof of the first point is obvious. We only prove the second point. From Proposition \ref{prop::fix_app} we deduce that
\begin{align*}
    dx_t = \sigma_t^2\frac{a^2}{t_1}(\frac{S^a_t+S^b_t}{2}-x_t)dt+\sigma_t^2a(dN^a_t-dN^b_t).
\end{align*}
Imposing $\frac{S^a_t+S^b_t}{2}=x_t$ and integrating we get that $x_t$ is a solution to
\begin{align*}
    x_t=x_0+\int_0^t \frac{a}{\frac{1}{\sigma_0^2}+a^2\frac{s}{t_1}}\big(dN^a_s-dN^b_s\big).
\end{align*}
\end{proof}
\comment{\begin{remark}
Under \hyperref[assump::fix]{$\Ac^\text{fix}$} and if $a(\frac{S^a+S^b}{2}-S_0)<<1$ then the limit of $x_t$ is approximately $S_0$. If $\frac{S^a_t+S^b_t}{2}=x_t$ and if $a(x_0-S_0)<<1$, we have
\begin{align*}
    \E[x_t] \simeq x_0+\int_0^t \frac{a}{\frac{1}{\sigma_0^2}+a^2\frac{s}{t_1}}(S_0-\E[x_t])
\end{align*}
so
\begin{align*}
    \E[x_t] \simeq \frac{\frac{t_1}{\sigma_0^2a^2}}{\frac{t_1}{\sigma_0^2a^2}+t}x_0+\frac{t}{\frac{t_1}{\sigma_0^2a^2}+t}S_0.
\end{align*}
\end{remark}}


\subsection{Brownian efficient price}
Assume now \hyperref[assump::sig]{$\Ac_\sigma$} with $\sigma>0$ instead of \hyperref[assump::sig]{$\Ac_0$}. We have the following result.

\begin{theorem}\label{prop::brow_approx}
Let $T_0>0$ and suppose that the market maker sets his quotes so that the mid-price $\frac{S^a_t+S^b_t}{2}$ is almost surely bounded on $[0,T_0]$ and constant between two jumps of $N^a$ or $N^b$ (this holds for example if $\frac{S^a_t+S^b_t}{2}=x_t$). There exists a unique positive solution $u$ to the approximate SPDE \eqref{eq::aprox_zakai} on $[0,T_0]\times\R$ such that $u_0=m_0$. It is given, for any $t\geq 0$ and up to a multiplicative constant\footnote{Recall that the solution of the approximate SPDE \eqref{eq::aprox_zakai} is not necessarily the density of a probability distribution.}, by the density of $\Nc(x_t,\sigma_t^2)$ where
\begin{align}\label{eq::sigmat}
\sigma_t^2 =
\begin{cases}
    \frac{\sigma \sqrt{t_1}}{a}\sqrt{1+\exp(-\frac{a\sigma}{2\sqrt{t_1}}t+C_0^+)} &\text{ if }\sigma_0^2\geq\frac{\sigma \sqrt{t_1}}{a}\\
    \frac{\sigma \sqrt{t_1}}{a}\sqrt{1-\exp(-\frac{a\sigma}{2\sqrt{t_1}}t+C_0^-)}&\text{ otherwise}
\end{cases}
\end{align}
with $C_0^+=\log(\sigma_0^4\frac{a^2}{\sigma^2 t_1}-1)$, $C_0^-=\log(1-\sigma_0^4\frac{a^2}{\sigma^2 t_1})$ and 
\begin{align}\label{eq::xt}
    dx_t = \sigma_t^2\frac{a^2}{t_1}(\frac{S^a_t+S^b_t}{2}-x_t)dt+\sigma_t^2a(dN^a_t-dN^b_t).
\end{align}
In particular, $\sigma_t^2$ converges to $\sigma^2_\infty=\frac{\sigma \sqrt{t_1}}{a}>0$ and a.s. $x_t$ does not converge.
\end{theorem}
\begin{proof}
We take an ansatz $u_t = \gamma_te^{-\frac{(x-x_t)^2}{2\sigma_t^2}}$ with $\sigma_t> 0$ and we look at the PDE which arises from \eqref{eq::aprox_zakai}, seen pathwise between $t=0$ and $t=\tau\wedge T_0$ where $\tau=\inf\{t>0,\max(N^a_t,N^b_t)>0\}$. Using \eqref{eq::aprox_zakai} and identifying the quadratic, linear and constant terms in $x$, we find that a solution of this PDE on $[0,\tau\wedge T_0)\times\R$ is given by a solution of the system
\begin{align*}
    &\partial_t\sigma_t = \frac{\sigma^2}{2}\frac{1}{\sigma_t}-\frac{a^2}{2t_1}\sigma_t^3\\
    &\partial_tx_t-\frac{2x_t}{\sigma_t}\partial_t\sigma_t=-x_t\frac{\sigma^2}{\sigma_t^2}+\sigma_t^2\frac{a^2}{t_1}\frac{S_t^a+S_t^b}{2}\\
    &\partial_t\gamma_t=\gamma_t\big(\frac{-1}{t_1}+\partial_t\big(\frac{x_t^2}{2\sigma_t^2}\big)-\frac{1}{2}\frac{a^2}{t_1}(\frac{S_t^a+S_t^b}{2})^2+\frac{\sigma^2}{2}(-\frac{1}{\sigma_t^2}+\frac{x_t^2}{\sigma_t^4})\big).
    \comment{&\partial_t\gamma_t=\gamma_t\big(\frac{-1}{t_1}+\frac{a^2}{t_1}x_t(\frac{3}{2}x_t+\frac{S_t^a+S_t^b}{2})-\frac{1}{2}\frac{a^2}{t_1}(\frac{S_t^a+S_t^b}{2})^2-\frac{\sigma^2}{2\sigma_t^2}\big)}
    \end{align*}
Given $x_0$, $\sigma_0$ and $\gamma_0=\frac{1}{\sqrt{2\pi\sigma_0^2}}$ a solution of this system exists and is unique. In particular $\sigma_t$ and $x_t$ are given by the \eqref{eq::sigmat} and \eqref{eq::xt} and $\gamma_te^{-\frac{(x-x_t)^2}{2\sigma_t^2}}$ is bounded on $[0,\tau\wedge T_0)\times\R$. From classical results on the positive solutions of the Cauchy problem with unbounded coefficients, as in \cite{Aronson1967UniquenessOP}, we deduce that this is the unique positive solution of \eqref{eq::aprox_zakai} on $[0,\tau\wedge T_0)\times\R$. If $\tau<T_0$ the solution stays Gaussian (up to a constant) after a jump at $\tau$. For example if $dN^a_\tau=1$ we have
\begin{align*}
    u_\tau(x) = u_{\tau-}(x)\lambda_0e^{-a(S^a_{\tau-}-x)}=\gamma_{\tau-}e^{-\frac{(x-x_{\tau-})^2}{2\sigma_{\tau-}^2}}\lambda_0e^{-a(S^a_{\tau-}-x)}=\gamma_{\tau}e^{-\frac{(x-x_{\tau})^2}{2\sigma_{\tau}^2}}
\end{align*}
with $\sigma_\tau=\sigma_{\tau-}$, $x_\tau=x_{\tau-}+a\sigma_{\tau-}^2$ and $\gamma_\tau=\gamma_{\tau_-}\lambda_0e^{-aS^a_{\tau-}}e^{-\frac{x_{\tau-}^2-(x_{\tau-}+a\sigma_{\tau-}^2)^2}{2\sigma_{\tau-}^2}}$.
Because there is almost surely a finite number of jumps on $[0,T_0]$, we deduce that a.s. we can extend this solution to $[0,T_0]\times\R$. In particular $\sigma_t$ must solve the first equation for all $t$, and $x_t$ satisfies
\begin{align*}
    dx_t=\big(\frac{2x_t}{\sigma_t}\partial_t\sigma_t-x_t\frac{\sigma^2}{\sigma_t^2}+\sigma_t^2\frac{a^2}{t_1}\frac{S^a_t+S^b_t}{2}\big)dt+\sigma_t^2 a(dN^a_t-dN^b_t).
\end{align*}
\end{proof}

This theorem shows that there is an asymptotic confidence in our estimation of the efficient price which comes from a balance between the diffusive behaviour of the price and the information gained by observing the trades. More precisely, the variance of our posterior is asymptotically constant and depends linearly on the (supposed) volatility of the efficient price. As a consequence in the approximation there is a permanent regime in which an observer's estimation of the price $x_t$ mean-reverts towards the mid-price at a speed $\frac{t_1}{a^2\sigma_\infty^2}$ and jumps by a fixed amount $\sigma^2_\infty a$ when a trade happens.

\begin{corollary}
If the market maker chooses the mid-price so that $\frac{S^a_t+S^b_t}{2}=x_t$ with fixed half-spread $\delta=\frac{S^a-S^b}{2}$, then the only solution is given by
\begin{align*}
    x_t=x_0+\int_0^t \sigma_s^2 a(dN^a_s-dN^b_s).
\end{align*}
\end{corollary}

\section{Impact of meta-orders on the posterior}\label{sec::meta}
We consider the case where, during the first $T\in\R\cup\{+\infty\}$ seconds, an additional market taker sends $\beta$ deterministic orders (buy orders if $\beta>0$ and sell orders if $\beta<0$) per second. In the following we suppose $\beta>0$ and we make the same assumptions as in Section \ref{sec::approx}: we assume \hyperref[assump::lamb]{$\Ac^\lambda$} and \hyperref[assump::sig]{$\Ac_\sigma$} for some $\sigma\geq 0$. The meta-order is treated as a perturbation of the process $N^a$. We assume that the observer does not make the difference between trades from a meta-order and trades from opportunistic market takers, and that he still uses the same filtering equation as in the case without meta-order. As a consequence we model the meta-orders by an additional term in the observed process $N^a$: $N^a_t$ is replaced by $N^a_t+N^\beta_t$ where $N^\beta$ is the deterministic c\`adl\`ag pure jump process defined by $N^\beta_0=0$ and $\Delta N^\beta_t=1$ if $t\beta\in\N^*$ and $t\leq T$, and $\Delta N^\beta_t=0$ otherwise. We use the same notations as in Section \ref{sec::approx}. Again, let $m_0$ be the density of a Gaussian distribution $\Nc(x_0,\sigma_0^2)$, and suppose that $\frac{S_t^a-S_t^b}{2}=\delta\geq 0$ for all $t$.

\subsection{Small spread approximation}
In the approximation introduced in Section \ref{sec::approx}, the filtering equation \eqref{eq::aprox_zakai} becomes
\begin{equation}\label{eq::aprox_zakai_meta}
    \begin{split}
    du_t(x) = &\frac{\sigma^2}{2}\partial^2_{xx}u_t(x)dt-\frac{1}{t_1}u_t(x)dt-\frac{1}{2t_1} a^2(x-\frac{S^a_t+S^b_t}{2})^2u_t(x)dt \\
    &+ u_{t-}(x)(\lambda_0 e^{-a(S_t^a-x)} - 1) (dN^a_t+dN^\beta_t)\\
    &+ u_{t-}(x)(\lambda_0 e^{-a(x-S_t^b)} - 1) dN^b_t.
    \end{split}
\end{equation}

From Proposition \ref{prop::fix_app}, Theorem \ref{prop::dirac} and Theorem \ref{prop::brow_approx} we deduce the following result. 
\begin{proposition} 
There exists a unique solution to \eqref{eq::aprox_zakai_meta} denoted by $u$ such that $u_0=m_0$. In particular, $u_t$ is (up to a multiplicative constant) the density of a Gaussian $\Nc(x_t,\sigma_t^2)$ where
\begin{align*}
    dx_t = \sigma_t^2\frac{a^2}{t_1}(\frac{S^a_t+S^b_t}{2}-x_t)dt+\sigma_t^2a(dN^a_t-dN^b_t+dN^\beta_t)
\end{align*}
and $\sigma_t$ is the same as in Proposition \ref{prop::fix_app} if $\sigma=0$ or Theorem \ref{prop::dirac} if $\sigma>0$.
\end{proposition}
We see that adding a meta-order does not change the dynamics of the variance of the posterior: only $x_t$ is impacted.\\

If we also impose that $\frac{S^a_t+S^b_t}{2}=x_t$ (with fixed half-spread $\delta=\frac{S^a-S^b}{2}$) then the solution is given by
\begin{align*}
    x_t=x_0+\int_0^t \sigma_s^2 a(dN^a_s-dN^b_s+dN^\beta_t).
\end{align*}
If in particular opportunistic market takers do not react to the meta-order ($N^a_t-N^b_t=0$ for $t\leq T$) then we have the following result.
\begin{corollary}\label{cor::imp}
If in addition $\frac{S^a_t+S^b_t}{2}=x_t$ (with fixed half-spread $\delta=\frac{S^a-S^b}{2}$) and $N^a_t-N^b_t=0$ for $t\leq T$, then, if $\sigma=0$, then
\begin{align*}
    x_t-x_0=\int_0^t\sigma_s^2 a dN^\beta_t\underset{\beta\rightarrow +\infty}{\sim}\frac{\beta t_1}{a}\log(1+\frac{a^2\sigma_0^2}{t_1}t) 
\end{align*}
for $t\leq T$. If $\sigma>0$, then
\begin{align*}
    x_t-x_0=\int_0^t\sigma_s^2 adN^\beta_t\underset{\beta\rightarrow +\infty}{\sim}&\frac{4\beta t_1}{a}\Bigg(-\sqrt{1\pm e^{-\frac{a\sigma}{2\sqrt{t_1}}t+C_0^\pm}}+\sqrt{1\pm e^{C_0^\pm}}\\
    &+\frac{1}{2}\log\left|\frac{1+\sqrt{1\pm e^{-\frac{a\sigma}{2\sqrt{t_1}}t+C_0^\pm}}}{1-\sqrt{1\pm e^{-\frac{a\sigma}{2\sqrt{t_1}}t+C_0^\pm}}}\right|-\frac{1}{2}\log\left|\frac{1+\sqrt{1\pm e^{C_0^\pm}}}{1-\sqrt{1\pm e^{C_0^\pm}}}\right|\Bigg)\\
    \underset{t\rightarrow +\infty}{\sim}&\beta\sqrt{t_1}\sigma t.
\end{align*}
\end{corollary}
So, believing that the efficient price follows a Brownian motion yields linear market impact if no information is given by the other market participants. In the extreme case where the observer does not expect the efficient price to move, logarithmic impact can appear. Hence our model does not reproduce the square-root law in the approximation for markets with small spreads. In this case logarithmic impact is a consequence of low volatility $\sigma$ of the efficient price.\\

Let us now consider the average behaviour of $x_t$. We see that
\begin{align}\label{eq::app_esp}
    d\E[x_t|S_0] \simeq -\frac{\sigma_s^2 a^2}{t_1}(\E[x_t|S_0]-S_0)dt+\sigma_t^2a dN^\beta_t,
\end{align}
in the small spread approximation, so 
\begin{align*}
    \E[x_t|S_0]\simeq \underbrace{x_0+(S_0-x_0)\int_0^t\frac{\sigma_s^2 a^2}{t_1}e^{-\int_s^t\frac{\sigma_u^2 a^2}{t_1}du}ds}_{\text{average learning: }A_t}+\underbrace{\int_0^t\sigma_s^2 a e^{-\int_s^t\frac{\sigma_u^2 a^2}{t_1}du}dN^\beta_s}_{\text{average impact: }B_t},
\end{align*}
where $A_t$ represents the average learning of the efficient price, and does not depend on the meta-order, and $B_t$ is a term of average impact caused by the meta-order.
We simulate the Zakai equation with $\frac{S^a_t+S^b_t}{2}=x_t$ and fixed half-spread $\delta=\frac{S^a-S^b}{2}$. A meta-order of size $25$ is executed. We plot the posterior's mean (minus $S_0$) as well as the approximate formula for $\E[x_t|S_0]-S_0$ given by \eqref{eq::app_esp} in the small spread approximation. We take the following parameters:
\begin{itemize}
    \item $\lambda_0=50$
    \item $\delta=0.1$
    \item $m_0\sim \Nc(S_0,\sigma_0^2)$ with $\sigma_0=0.05$, so that $A_t=x_0$
    \item $\sigma = 0.06$.
\end{itemize}
We take $a=5$ (small spreads compared to the characteristic scale $1/a$ of the intensity function) in Figure \ref{fig:impact5} and $a=20$ (large spreads compared to the characteristic scale $1/a$ of the intensity function) in Figure \ref{fig:impact20}. We vary the speed $\beta$ of the meta-order. The full lines correspond to a simulation of $x_t$, while the dashed lines correspond to \eqref{eq::app_esp}.

\begin{figure}[h]
    \centering
    \includegraphics[width = 18cm]{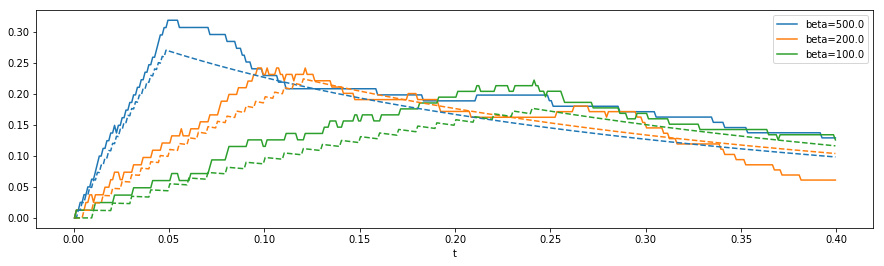}
    \caption{Simulation of the price impact with respect to time, for $a=5$ and three choices of $\beta$. The dotted plot is the theoretical mean impact in the small spread approximation.}
    \label{fig:impact5}
\end{figure}

\begin{figure}[h]
    \centering
    \includegraphics[width = 18cm]{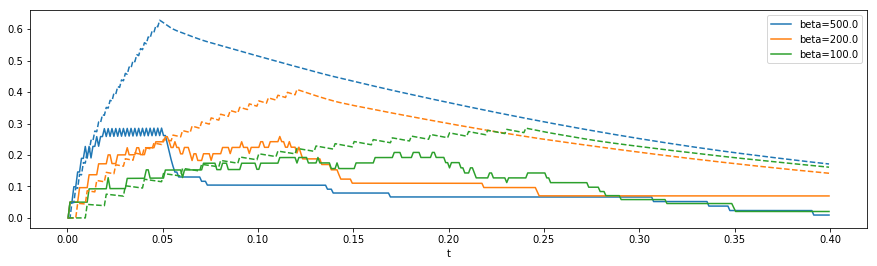}
    \caption{Simulation of the price impact with respect to time, for $a=20$ and three choices of $\beta$. The dotted plot is the theoretical mean impact in the small spread approximation.}
    \label{fig:impact20}
\end{figure}
As expected, we observe that the small-spread approximation works well only for a small $a$, and that in this case the impact is roughly linear in time.\\

Next we note that $A_t$ has closed form formulas if $\sigma=0$ or if $\sigma>0$ and $\sigma_0=\sigma_\infty$. The term $B_t$ can also be expressed in a simple way in the limit $\beta\rightarrow+\infty$.
\begin{proposition}\label{avimp}
For $t\geq 0$, we have
\begin{align*}
    A_t=
    \begin{cases}
        \frac{\frac{t_1}{\sigma_0^2a^2}}{\frac{t_1}{\sigma_0^2a^2}+t}x_0+\frac{t}{\frac{t_1}{\sigma_0^2a^2}+t}S_0&\text{if }\sigma=0,\\
        x_0+(S_0-x_0)\int_0^t\frac{\sigma_s^2 a^2}{t_1}e^{-(t-s)\frac{a\sigma}{\sqrt{t_1}}}ds\underset{\text{if }\sigma_0=\sigma_\infty}{=}S_0+(x_0-S_0)e^{-\sigma_\infty^2\frac{a^2}{t_1}t}& \text{if }\sigma>0,
    \end{cases}
\end{align*}
and 
\begin{align*}
    B_t \underset{\beta\rightarrow +\infty}{\sim} \begin{cases}
        \frac{t\wedge T}{\frac{t_1}{\sigma_0^2a^2}+t}\frac{\beta t_1}{a}& \text{if }\sigma=0,\\
        \int_0^{t\wedge T}\sigma_s^2 a\beta e^{-(t-s)\frac{a\sigma}{\sqrt{t_1}}}dt\underset{\text{if }\sigma_0=\sigma_\infty}{=}\frac{\beta t_1}{a}e^{-\sigma_\infty^2\frac{a^2}{t_1}t}(e^{\sigma_\infty^2\frac{a^2}{t_1}(t\wedge T)}-1)& \text{if }\sigma>0.
    \end{cases}
\end{align*}
\end{proposition}
We see that as long as the market takers keep sending orders with an intensity that depends on the efficient price, we can expect the mid-price to reflect the efficient price (via $A_t$) plus some impact (via $B_t$). The expected mid-price $x_t$ is the sum of two terms: one term $A_t$ which is left unchanged by a meta-order and which denotes the learning of the efficient price allowed by the information sent by the market takers, and another term $B_t$ which reflects the impact of the meta-order on the mid-price and behaves as a concave function of time. Interestingly if we look only at the filtering equation, then imposing $N^a-N^b=0$ has two different interpretations: it can mean that the time period we are looking at is small or that the buys equal the sells \textit{i.e.} that the market takers trade equally on both sides, \textit{i.e.} that on this sample path $S_t$ is close to $x_t$. So it reflects the case in which the market takers learn similarly to the market maker. In that case the impact becomes linear in the long run, except if the market maker is convinced that the efficient price does not move, in which case the impact is logarithmic. Note that the two choices we make for the market takers (averaging trades with some fixed $S$ or taking $N^a-N^b=0$) are actually two extreme cases where the price formation process is driven either by market takers who are price arbitragers, or by the market maker. The study of a transition between those two cases, which should take into account the views of the market takers, is left for further study.

\subsection{No approximation}
Now we do not approximate the Zakai SPDE by \eqref{eq::aprox_zakai} anymore. We assume \hyperref[assump::sig]{$\Ac_0$}. In the presence of a meta-order, our (perturbed) filtering equation \eqref{eq::zakai_spde} becomes 
\begin{equation}\label{eq::zakai_meta}
    \begin{split}
    du_t(x) = &-\frac{1}{t_1}\cosh(a(x-\frac{S^a_t+S^b_t}{2}))u_t(x)dt \\
    &+ u_{t-}(x)(\lambda_0 e^{-a(S_t^a-x)} - 1) (dN^a_t+dN^\beta_t)\\
    &+ u_{t-}(x)(\lambda_0 e^{-a(x-S_t^b)} - 1) dN^b_t.  
    \end{split}
\end{equation}
The posterior does not remain Gaussian. In Section \ref{sec::approx}, we made the market maker take for his mid-price the mean of the Gaussian posterior, \textit{i.e.} the likelihood maximizer. Here the mean and the likelihood maximizer are not necessarily equal, but the likelihood maximizer is easier to manipulate and more natural. If it is well defined, we denote it by $\hat x_t$:
\begin{align*}
    \hat x_t = \text{argmax}\,u_t.
\end{align*}
We will see that the filtering equation \eqref{eq::zakai_meta} leads to some new interesting dynamics for $\hat x_t$, both in the case with fixed mid-price and in the case with a mid-price equal to $\hat{x}_t$. In particular it naturally gives rise to an ``arcsinh impact law''.

\subsubsection{Fixed prices}
As in Section \ref{sec::approx} we start with the simple case of fixed bid and ask prices. We have the following result on $\hat x_t$, which says that we can expect an impact in $\textnormal{arcsinh}(\beta)$ as the starting confidence (the standard deviation $\sigma_0$ of the initial prior) becomes small ($\sigma_0\rightarrow+\infty$) or the quantity of information becomes large ($t\rightarrow +\infty$).
\begin{proposition}
Assume \hyperref[assump::fix]{$\Ac^\text{fix}$}. 
Then, $\hat x_t$ is well defined for all $t\geq 0$, and, a.s., 
\begin{align*}
    \hat x_t\underset{t\rightarrow+\infty}{\longrightarrow}\begin{cases}
    S&\textnormal{ if }T<+\infty,\\
    \frac{S^a+S^b}{2}+\frac{1}{a} \textnormal{arcsinh}(\sinh(a(S-\frac{S^a+S^b}{2}))+\beta t_1)&\textnormal{ if }T=\infty.
    \end{cases}
\end{align*}
Also, for all $t\geq 0$, a.s.,
\begin{align*}
    \hat x_t \underset{\sigma_0\rightarrow+\infty}{\longrightarrow}\frac{S^a+S^b}{2}+\frac{1}{a}\textnormal{arcsinh}\big(\frac{t_1}{t}(N^a_t-N^b_t+\lfloor\beta (t\wedge T)\rfloor)\big).
\end{align*}
\end{proposition}
\begin{proof}
We prove without loss of generality the convergence in $t$. Using Proposition \ref{prop_simple} we get
\begin{align*}
    \hat x_t &= \text{argsup}\{e^{-\frac{(x-x_0)^2}{2\sigma_0^2}-\frac{t}{t_1}\cosh(a(\frac{S^a+S^b}{2}-x))+ax(N^a_t-N^b_t)+ax\lfloor\beta (t\wedge T)\rfloor}\}\\
    &= \text{arginf}\{\frac{(x-x_0)^2}{2t\sigma_0^2}+\frac{1}{t_1}\cosh(a(\frac{S^a+S^b}{2}-x))-ax\frac{N^a_t-N^b_t}{t}-ax \frac{\lfloor\beta (t\wedge T)\rfloor}{t}\}.
\end{align*}
The function $x \mapsto \frac{(x-x_0)^2}{2\sigma_0^2}+\frac{t}{t_1}\cosh(a(\frac{S^a+S^b}{2}-x))$ is strictly convex for every $t$ and goes to $-\infty$ as $x\rightarrow \pm \infty$ so $\hat x_t$ is well defined and is an argmax. Almost surely the function $x\mapsto \frac{(x-x_0)^2}{2t\sigma_0^2}+\frac{1}{t_1}\cosh(a(\frac{S^a+S^b}{2}-x))-ax\frac{N^a_t-N^b_t}{t}-ax \frac{\lfloor\beta (t\wedge T)\rfloor}{t}$ is level-bounded in $x$ uniformly in $t$ and converges uniformly in $x$ on every compact set of $\R$ so epiconverges towards the function $x\mapsto \frac{1}{t_1}\cosh(a(\frac{S^a+S^b}{2}-x))-\frac{a}{t_1}\sinh(a(S-\frac{S^a+S^b}{2}))x-ax\beta\mathbf{1}_{T=+\infty}$ as $t\rightarrow +\infty$. As a consequence of classical results \cite{varanalysis}, the argmin of the first function converges to the argmin of the second function, which is 
\begin{align*}
    \frac{S^a+S^b}{2}+\frac{1}{a}\text{arcsinh}(\sinh(a(S-\frac{S^a+S^b}{2})))=S
\end{align*}
for $T<+\infty$ and 
\begin{align*}
    \hat x_t\rightarrow \frac{S^a+S^b}{2}+\frac{1}{a}\text{arcsinh}(\sinh(a(S-\frac{S^a+S^b}{2}))+\beta t_1)
\end{align*}
otherwise.
\end{proof}

The boundedness of the impact can be explained by the fact that the market marker's quotes remain constant, so at some point there is a balance between the increase in the estimation of $S$ due to the meta-order and the decrease towards the mid-price set by the market maker between the trades.

\subsubsection{With a market maker}
We now look for a solution in the case where the market maker sets his mid-price as $\hat x_t$. We have $\hat x_0=x_0$ and, for any $t>0$, $\hat x_t$ should solve
\begin{align}\label{argmax}
    \hat x_t = \textnormal{argmax}\big\{-\frac{(x-x_0)^2}{2t\sigma_0^2}-\frac{1}{t_1}\frac{1}{t}\int_0^t\cosh(a(\hat x_s-x))ds+ax\frac{N^a_t-N^b_t}{t}+ax\frac{\lfloor\beta (t\wedge T)\rfloor}{t}\big\}.
\end{align}
We will obtain explicit formulas for the impact in two limiting regimes. If $\beta t_1$ is large (\textit{i.e.} there are few opportunistic trades compared to meta-orders), we obtain an impact which is logarithmic in the traded volume. In the case of a very slow meta-order and under some approximations we find constant impact. In the intermediate regime the impact is closer to an arcsinh of the traded volume, at least for small volumes.\\

We start by showing that $\hat x_t$ is necessarily constant between two jumps.
\begin{proposition}
There is a unique c\`adl\`ag finite variation process $\hat x_t$ solution to \eqref{argmax} with $\hat x_0=x_0$. It is given by a pure jump process with jumps happening only when $dN^a_t\neq 0$ or $dN^b_t\neq 0$ or $\{t|\beta|\in\N^*,t\leq T\}$ and the jump magnitudes $\Delta \hat x_t=\hat x_t-\hat x_{t-}$ solve
\begin{align*}
    0 = &-\frac{\hat x_{t-}+\Delta\hat x_t-x_0}{\sigma_0^2}-\frac{a}{2t_1}\big(e^{a(\hat x_{t-}+\Delta\hat x_t)}\int_0^te^{-a\hat x_s}ds-e^{-a(\hat x_{t-}+\Delta\hat x_t)}\int_0^te^{a\hat x_s}ds\big)\\
    &+a\int_0^t(dN^a_s-dN^b_s+dN^\beta_t)
\end{align*}
if $dN^a_t\neq 0$ or $dN^b_t\neq 0$ or $\{t\beta\in\N^*,t\leq T\}$.
\end{proposition}
\begin{proof}
Given \eqref{argmax} $\hat x_t$ solves
\begin{align}\label{eq::genhatx}
    0 = -\frac{\hat x_t-x_0}{\sigma_0^2}-\frac{a}{2t_1}\big(e^{a\hat x_t}\int_0^te^{-a\hat x_s}ds-e^{-a\hat x_t}\int_0^te^{a\hat x_s}ds\big)+a\int_0^t(dN^a_s-dN^b_s+dN^\beta_t).
\end{align}
Differentiating we get, by denoting by $d\hat x_t$ and $\Delta\hat x_t$ the continuous and jump differentials,
\begin{align*}
    0= &-\frac{d\hat x_t+\Delta \hat x_t}{\sigma_0^2}-\frac{a^2}{2t_1}\big(e^{a\hat x_{t-}}\int_0^te^{-a\hat x_s}ds+e^{-a\hat x_{t-}}\int_0^te^{a\hat x_s}ds\big)d\hat x_t\\
    &-\frac{a}{2t_1}\big(e^{a(\hat x_{t-}+\Delta\hat x_t)}\int_0^te^{-a\hat x_s}ds-e^{-a(\hat x_{t-}+\Delta\hat x_t)}\int_0^te^{a\hat x_s}ds\big)\\
    &+\frac{a}{2t_1}\big(e^{a\hat x_{t-}}\int_0^te^{-a\hat x_s}ds-e^{-a\hat x_{t-}}\int_0^te^{a\hat x_s}ds\big)\\
    &+a(dN^a_t-dN^b_t+dN^\beta_t).
\end{align*}
So identifying the jump and continuous parts we get the desired result.
\end{proof}
This extends Proposition \ref{prop::stabil} as it shows that the market maker has no incentive to change his mid-price if there is no trade and if the mid-price maximizes the density function of his prior. Also, it gives an explicit recursive equation to compute the size of the jumps in the mid-price. Surprisingly this equation is one-dimensional: we do not need to make computations on density functions to solve it. 

\paragraph{Fast meta-order} In the limit where the learning is fast and the initial prior has large variance we get the following corollary when we look at the event $\{N^a_t=N^b_t\text{ for all }t<T\}$.

\begin{corollary}
Assume $\{N^a_t=N^b_t\text{ for all }t<T\}$. There is a unique c\`adl\`ag finite variation process $\hat x^0_t$ solution to \eqref{argmax} with $\hat x^0_0=x_0$. On $[0,T]$ it is a deterministic pure jump-process with jumps happening only if $t|\beta|\in\N^*$ and $0\leq t\leq T$, and the jump magnitudes $\Delta \hat x^0_t=\hat x^0_t-\hat x^0_{t-}$ solve
\begin{align*}
    0 = &-\frac{\hat x^0_{t-}+\Delta\hat x^0_t-x_0}{\sigma_0^2}-\frac{a}{2t_1}\big(e^{a(\hat x^0_{t-}+\Delta\hat x^0_t)}\int_0^te^{-a\hat x^0_s}ds-e^{-a(\hat x^0_{t-}+\Delta\hat x^0_t)}\int_0^te^{a\hat x^0_s}ds\big)+a\lfloor\beta t\rfloor.
\end{align*}
Also, we have 
\begin{align*}
    \hat x^0_T-x_0 \sim \frac{1}{a}\log(2\lfloor\beta T\rfloor\beta t_1)
\end{align*}
as $\frac{a^2\sigma_0^2}{\beta t_1}\rightarrow+\infty$ and $\beta t_1\rightarrow+\infty$ with $a$ and $\beta$ fixed.
\end{corollary}
\begin{proof}
As we work on $\{N^a_t=N^b_t\text{ for all }t<T\}$ we can build $\hat x_t^0$ recursively. For $k\in\{0,...,\lfloor T\beta\rfloor\}$,
\begin{align*}
    ka = \frac{\hat x^0_{\frac{k}{\beta}}-x_0}{\sigma_0^2}+\frac{a}{t_1}\int_0^{\frac{k}{\beta}}\sinh(a(\hat x^0_{\frac{k}{\beta}}-\hat x^0_s))ds
\end{align*}
so
\begin{align}\label{eq::rec_impact}
    k = \frac{\hat x^0_{\frac{k}{\beta}}-x_0}{a\sigma_0^2}+\frac{1}{\beta t_1}\sum_{i=0}^{k-1}\sinh(a(\hat x^0_{\frac{k}{\beta}}-\hat x^0_{\frac{i}{\beta}})).
\end{align}
We get a recursive sequence with a fixed number of steps. We show the convergence for the first two steps ($k=1$ and $k=2$). The same argument can be used for any $k\in\{0,...,\lfloor T\beta\rfloor\}$.
First
\begin{align*}
    1 = \frac{\hat x^0_{\frac{1}{\beta}}-x_0}{a\sigma_0^2}+\frac{1}{\beta t_1}\sinh(a(\hat x^0_{\frac{1}{\beta}}-x_0))
\end{align*}
and as $\frac{a^2\sigma_0^2}{\beta t_1}\rightarrow+\infty$
\begin{align*}
    1 = \frac{1}{\beta t_1}\sinh(a(\hat x^0_{\frac{1}{\beta}}-x_0))+o(1),
\end{align*}
and using $\beta t_1\rightarrow+\infty$ we get
\begin{align*}
    1\sim\frac{1}{2\beta t_1}e^{(a(\hat x^0_{\frac{1}{\beta}}-x_0))}
\end{align*}
hence 
\begin{align*}
    \hat x^0_{\frac{1}{\beta}}-x_0= \frac{1}{a}\log(2\beta t_1)+o(1).
\end{align*}
Then 
\begin{align*}
    2 = \frac{\hat x^0_{\frac{2}{\beta}}-x_0}{a\sigma_0^2}+\frac{1}{\beta t_1}\sinh(a(\hat x^0_{\frac{2}{\beta}}-\hat x^0_{\frac{1}{\beta}}))+\frac{1}{\beta t_1}\sinh(a(\hat x^0_{\frac{2}{\beta}}-x_0)).
\end{align*}
By the same arguments
\begin{align*}
    2=\frac{1}{2\beta t_1}e^{a(\hat x^0_{\frac{2}{\beta}}-x_0)}+\frac{1}{\beta t_1}\sinh(a(\hat x^0_{\frac{2}{\beta}}-\hat x^0_{\frac{1}{\beta}}))+o(1)
\end{align*}
as $\beta t_1 \rightarrow+\infty$. If $a(\hat x^0_{\frac{2}{\beta}}-\hat x^0_{\frac{1}{\beta}})$ is not bounded as $\beta t_1 \rightarrow+\infty$, then along some sequence such that $(\beta t_1)_k\underset{k\rightarrow +\infty}{\rightarrow}+\infty$ and $a(\hat x^0_{\frac{2}{\beta}}-\hat x^0_{\frac{1}{\beta}})\underset{k\rightarrow +\infty}{\rightarrow}+\infty$ we have
\begin{align*}
    2=\frac{1}{2(\beta t_1)_k}e^{a(\hat x^0_{\frac{2}{\beta}}-x_0)}+\frac{1}{2(\beta t_1)_k}e^{a(\hat x^0_{\frac{2}{\beta}}-\hat x^0_{\frac{1}{\beta}})}+o(1),
\end{align*}
but 
\begin{align*}
    \frac{e^{a(\hat x^0_{\frac{2}{\beta}}-\hat x^0_{\frac{1}{\beta}})}}{e^{a(\hat x^0_{\frac{2}{\beta}}-x_0)}}=e^{-a(\hat x^0_{\frac{1}{\beta}}-x_0)}\underset{k\rightarrow+\infty}{\rightarrow} 0,
\end{align*}
so the second term is negligible and 
\begin{align*}
    \hat x^0_{\frac{2}{\beta}}-x_0= \frac{1}{a}\log(4\beta t_1)+o(1),
\end{align*}
which contradicts that $a(\hat x^0_{\frac{2}{\beta}}-\hat x^0_{\frac{1}{\beta}})\underset{k\rightarrow +\infty}{\rightarrow}+\infty$. So $a(\hat x^0_{\frac{2}{\beta}}-\hat x^0_{\frac{1}{\beta}})$ is bounded as $\beta t_1 \rightarrow+\infty$. As a consequence
\begin{align*}
    2=\frac{1}{2\beta t_1}e^{a(\hat x^0_{\frac{2}{\beta}}-x_0)}+o(1),
\end{align*}
which again yields 
\begin{align*}
    \hat x^0_{\frac{2}{\beta}}-x_0\sim \frac{1}{a}\log(4\beta t_1).
\end{align*}
\end{proof}

This means that if a meta-order takes place and the market maker adapts his quotes accordingly without benefiting from any additional information given by another trader, the mid-price increases logarithmically in time (or in the traded volume). This is a bit different from Corollary \ref{cor::imp} with $\sigma=0$ as now the trading speed $\beta$ appears only in the total traded volume in the logarithm.\\

\paragraph{Slow meta-order} We now consider the other limit regime where $\beta t_1<<1$ and $T=+\infty$ \textit{i.e.} the meta-order is very slow. We make some simplifications. In \eqref{eq::genhatx} we formally take $\sigma_0\rightarrow+\infty$ and we replace the jump processes by their compensators and $N^\beta_t$ by $\beta t$. Setting $te^{aS_0}u(t)=\int_0^t e^{a\hat x_s}ds$ and $te^{-aS_0}v(t)=\int_0^t e^{-a\hat x_s}ds$ with $u$ and $v$ two differentiable functions on $\R_+$ and solving in $e^{a\hat x_t}=e^{aS_0}u(t)+te^{aS_0}u'(t)$ and in $e^{-a\hat x_t}=e^{-aS_0}v(t)+te^{-aS_0}v'(t)$ leads us to the system of equations
\begin{align}
    \begin{split}\label{sys1}
    \frac{1}{t_1}(uv+tu'v)&=(\beta+\frac{1}{2t_1}(v-u))+\sqrt{(\beta+\frac{1}{2t_1}(v-u))^2+\frac{uv}{t_1^2}},\\
    \frac{1}{t_1}(uv+tuv')&=-(\beta+\frac{1}{2t_1}(v-u))+\sqrt{(\beta+\frac{1}{2t_1}(v-u))^2+\frac{uv}{t_1^2}}
\end{split}
\end{align}

on $\R_+^*$. Taking $\beta t_1<<1$ leads us to consider the approximate problem
\begin{align}
    \begin{split}\label{sys2}
    \frac{1}{t_1}(uv+tu'v)&=(\beta+\frac{1}{2t_1}(v-u))+\frac{1}{2t_1}(v+u)+\beta\frac{v-u}{v+u},\\
    \frac{1}{t_1}(uv+tuv')&=-(\beta+\frac{1}{2t_1}(v-u))+\frac{1}{2t_1}(v+u)+\beta\frac{v-u}{v+u}
\end{split}
\end{align}
on $\R_+^*$. 
Solving this system leads to the following property.

\begin{proposition}\label{prop::slow}
The system \eqref{sys2} has a unique solution given by
\begin{align*}
    u(t) = 1+\beta t_1, v(t)=1-\beta t_1.
\end{align*}
\end{proposition}
\begin{proof}
See Appendix \ref{app::proof_slow}.
\end{proof}

In particular, going back to $\hat x$, we get $\hat x_t=S+\frac{\log(1+\beta t_1)}{a}$ from $u$ and $\hat x_t=S-\frac{\log(1-\beta t_1)}{a}$ from $v$, which is reasonable and tells us that the impact is finite and equal to $\frac{\beta t_1}{a}$.\\

\paragraph{Intermediate regime} The recursive formula \eqref{eq::rec_impact} can also be used to compute recursively the theoretical market impact in the case where $\beta t_1$ is fixed and $\sigma_0\rightarrow+\infty$, but the formulas are very hard to use. For example, if $N^a_t=N^b_t$ for all $t\leq T$:
\begin{align*}
    &\hat x_{\frac{1}{\beta}}-x_0\underset{\sigma_0\rightarrow+\infty}{\sim} \frac{1}{a}\textnormal{arcsinh}(\beta t_1)\\
    &\hat x_{\frac{2}{\beta}}-x_0\underset{\sigma_0\rightarrow+\infty}{\sim} \frac{1}{a}\textnormal{arcsinh}\big(\beta t_1(1+\sqrt{1-\frac{3}{2}\frac{1}{1+\sqrt{1+(\beta t_1)^2}}})\big).
\end{align*}
The impact of $Q$ orders at the speed $\beta$ per second takes the shape $\frac{1}{a}\textnormal{arcsinh}(f(Q)\beta t_1)$ for some sublinear function $f$, which is close to linear for small $Q$ and large $\beta t_1$.\\

\paragraph{Summary of the results on market impact in this section}In this section we have shown how various market impact shapes arise from meta-order splitting depending on the information structure, the market parameters and the nature of the meta-orders. If the spread is small (\textit{i.e.} in the approximation with $a$ small and $\sigma_0$ small) and the mid-price is fixed as the average of the posterior, we derive two types of impact depending on who drives the market. If the market is driven by the market maker ($N^a-N^b=0$), the market impact is linear if $\sigma>0$ while it is logarithmic in time and linear in the speed $\beta$ in the extreme case $\sigma=0$. If opportunistic market takers are present, the average market impact grows as $\frac{t\wedge T}{\frac{t_1}{\sigma_0^2a^2}+t}\frac{\beta t_1}{a}$ if $\sigma>0$ and as $\frac{\beta t_1}{a}e^{-\sigma_\infty^2\frac{a^2}{t_1}t}(e^{\sigma_\infty^2\frac{a^2}{t_1}(t\wedge T)}-1)$ if $\sigma=0$.\\

We also compute market impact shapes without using the approximation if the mid-price is fixed as the average of the posterior and $\sigma=0$. If $N^a-N^b=0$ the impact is logarithmic in the volume $\beta t$ if $\sigma_0^2/t_1\rightarrow+\infty$ and $t_1\rightarrow+\infty$. If we only take $\sigma_0\rightarrow+\infty$ we obtain a recursive formula in which the key quantity to compute is the $sinh$ of the impact.

\bibliographystyle{apalike}
\bibliography{biblio.bib}

\appendix
\section{Proof of Proposition \ref{prop::lamb}}\label{app:proof_prop_lamb}
\paragraph{Property (a)}
By a symmetry argument we prove the result for a trade on the ask side only.
Let $l$ be a deterministic, positive, non-increasing and exponentially bounded function. $\Mc^a_l$ is non-empty as it contains the Gaussian density functions. Let $m\in\Mc^a_l$.\\

The image of $m$ by the map is the function defined by $$\tilde m (z,x) = \frac{l(z-x)m(x)}{\int l(z-y)m(y) dy}$$ and is well-defined and positive by the definition of $\Mc^a_l$. Suppose $\tilde m$ does not depend on $z$. From $m>0$ and $\tilde m (z,x)=\tilde m (0,x)$, we get
$$l(z-x) = f(z)l(-x)$$
for any $x\in\R$, $z\in\R$ and for a positive function $f$ which is given by $f(z)=\frac{\int l(z-y)m(y) dy}{\int l(-y)m(y) dy}$ and does not depend on $x$. Taking $x=0$ we get $l(z)=f(z)l(0)$ so
$$l(0)l(z-x) = l(z)l(-x)$$
for all $(y,x)\in\R^2$. Using the monotonicity of $l$ we have $l\in\{x\mapsto \lambda_0 e^{-ax}, (\lambda_0,a)\in\R^2_+\}$. The other inclusion is straightforward.

\paragraph{Property (b)}
We prove the first inclusion, as the other one is obvious.
We look for functions $l$ which are continuous, strictly decreasing, exponentially bounded, convex and four times differentiable, such that there exist functions $f,g$ on $\R$ and $h$ on $\R^2$ such that, for any $x,s_a,s_b\in\R$,
\begin{equation*}
    \begin{split}
    -(l(s^a-x) +l(x-s^b) - 2) =h(s^a,s^b)-g(x-\frac{s^a+s^b}{2})f(\frac{s^a-s^b}{2}).
    \end{split}
\end{equation*}
However up to renaming we can take $\min g = g(0) = 0$ so $h(s^a,s^b)=-2(l(\frac{s^a-s^b}{2})-1)$ and
\begin{equation*}
    \begin{split}
    -(l(s^a-x) +l(x-s^b)) =-2l(\frac{s^a-s^b}{2}) -g(x-\frac{s^a+s^b}{2})f(\frac{s^a-s^b}{2}).
    \end{split}
\end{equation*}
We can suppose that $l^{(4)}(0)>0$. Indeed as $l$ converges and is decreasing $l'$ converges to $0$ as it is increasing and bounded. We can find some point $x$ such that $l^{(3)}(x)<0$. Indeed if $l^{(3)}\geq 0$, as $l''(z)>0$ for some $z\in\R$, $l$ is strongly convex on $[z,+\infty[$, so diverges, which contradicts the convergence of $l$. Then if $l^{(4)}(y)\leq 0$ for all $y\in\R$ then $l^{(3)}(y)\leq l^{(3)}(x)$ for all $y\geq x$ so $l'$ is strongly concave on $[x,+\infty[$, so diverges, which is absurd.\\

First we note that $f(x)\neq 0$ for all $x\in\R$. So, by taking $s^a=s^b=0$ and $s^a=-s^b$ we can write
\begin{equation*}
    \begin{split}
    l(y-x) +l(y+x)-2l(y)= (l(x) +l(-x)-2l(0))\frac{f(y)}{f(0)}
    \end{split}
\end{equation*}
for any $x,y\in\R$. Differentiating twice in the $x$ direction and taking $x=0$ we find $l''(x)=l''(0)\frac{f(x)}{f(0)}$ and 
\begin{equation*}
    \begin{split}
    l''(y-x) +l''(y+x)= (l''(x) +l''(-x))\frac{l''(y)}{l''(0)}.
    \end{split}
\end{equation*}
By differentiating again twice with respect to $x$ and taking $x=0$ we get
\begin{equation*}
    \begin{split}
    l^{(4)}(y)= l^{(4)}(0)\frac{l''(y)}{l''(0)},
    \end{split}
\end{equation*}
which yields the result given that $l$ is positive and strictly decreasing.

\section{Proof of Theorem \ref{prop::dirac}}\label{app::proof_dirac}
\comment{Let $X$ be a random variable associated with a probability measure steady state of the KS equation. The KS equation says that for any $f\in\C^2_c(\R)$, $Cov(f(X),\lambda(S^a-X)+\lambda^b(X-S^b))=0$. $\lambda(S^a-.)+\lambda(.-S^b)$ is increasing on $[\frac{S^a+S^b}{2},\infty)$. If we take $f$ increasing with support in $[\frac{S^a+S^b}{2},\infty)$ and use the equality case for the theorem about the covariance monotone functions, we get that $f$ and $\lambda(S^a-.)+\lambda(.-S^b)$ are colinear on any Borel set $A$ included in $[\frac{S^a+S^b}{2},\infty)$ such that $\P(X\in A) = \P(X\geq \frac{S^a+S^b}{2})$. As $f$ can be chosen arbitrarily, $X$ has at most one atom on $[\frac{S^a+S^b}{2},\infty)$ and no diffuse part. Doing the same thing on the other side we get that $X$ is composed of 2 atoms at most. Doing the same reasoning for those two atoms we find that they are symmetrical.\\
\\}
Let $m_0\in\Cc(\R)$ be the density of a probability measure on $\R$ with $m_0(\frac{S^a+S^b}{2})\neq 0$.
\paragraph{Proof of \textit{(i)}.} If $t<\tau$, then
\begin{align*}
    \hat m_t(x) &= m_0(x)e^{2t-\int_0^t(\lambda(S^a_u-x)+\lambda(x-S^b_u))du}\\
     &= m_0(x)e^{2t-t(\lambda(S^a_u-x)+\lambda(x-S^b_u))},
\end{align*}
which is integrable because $\lambda$ is non-negative. Now note that $\Phi_t: x\mapsto e^{-t(\lambda(S^a_u-x)+\lambda(x-S^b_u))}$ is strictly increasing on $]-\infty,\frac{S^a+S^b}{2}]$ and strictly decreasing on $[\frac{S^a+S^b}{2},+\infty[$. We look at the renormalized density $u_t$. At $y = \frac{S^a+S^b}{2}$, we have $m_0(y)\neq 0$ by the assumption, and 
\begin{align*}
    u_t(y) &= \frac{m_0(y)\Phi_t(y)}{\int m_0(x)\Phi_t(x)dx}= \frac{1}{\int \frac{m_0(x)}{m_0(y)}\frac{\Phi_t(x)}{\Phi_t(y)}dx}.
\end{align*}
Now $m_0$ is integrable and $\frac{\Phi_t(x)}{\Phi_t(y)}\longrightarrow 0$ for any $x\neq y$. By dominated convergence theorem we deduce that $u_t(y)\longrightarrow +\infty$.\\
Now take $y\neq \frac{S^a+S^b}{2}$ such that $m_0(y)\neq 0$. Without loss of generality we suppose $y<\frac{S^a+S^b}{2}$. Observe that for any $t$, $\Phi_t$ is symmetric around $\frac{S^a+S^b}{2}$. So 
\begin{align*}
    u_t(y) &= \frac{m_0(y)\Phi_t(y)}{\int m_0(x)\Phi_t(x)dx}\\
    &= \frac{1}{\int \frac{m_0(x)}{m_0(y)}\frac{\Phi_t(x)}{\Phi_t(y)}dx}\\
    & = \frac{1}{\int_{-\infty}^y \frac{m_0(x)}{m_0(y)}\frac{\Phi_t(x)}{\Phi_t(y)}dx+\int_y^{S^a+S^b-y} \frac{m_0(x)}{m_0(y)}\frac{\Phi_t(x)}{\Phi_t(y)}dx+\int_{S^a+S^b-y}^{+\infty} \frac{m_0(x)}{m_0(y)}\frac{\Phi_t(x)}{\Phi_t(y)}dx}.
\end{align*}
By the same reasoning as before the first and third terms in the denominator tend to $0$. Now note that there exists a closed interval of the form $[\frac{S^a+S^b}{2}-\epsilon,\frac{S^a+S^b}{2}+\epsilon]$ in $]y,S^a+S^b-y[$ with $\epsilon>0$ and $m_0(x)>\frac{m_0(\frac{S^a+S^b}{2})}{2}$ for any $x$ in this interval. So
\begin{align*}
    \int_y^{S^a+S^b-y} \frac{m_0(x)}{m_0(y)}\frac{\Phi_t(x)}{\Phi_t(y)}dx&\geq \int_{\frac{S^a+S^b}{2}-\epsilon}^{\frac{S^a+S^b}{2}+\epsilon} \frac{m_0(x)}{m_0(y)}\frac{\Phi_t(x)}{\Phi_t(y)}dx\\
    &\geq 2\epsilon\frac{m_0(\frac{S^a+S^b}{2})}{2m_0(y)}\frac{\Phi_t(\frac{S^a+S^b}{2}-\epsilon)}{\Phi_t(y)}\longrightarrow +\infty
\end{align*}
and 
\begin{align*}
    u_t(y)\longrightarrow 0.
\end{align*}

\paragraph{Proof of \textit{(ii)}}
Let $y\in\R$. Observe that
\begin{align*}
    u_t(y)& = \frac{m_0(y)\Phi_t(y)}{\int m_0(x)\Phi_t(x)}=\frac{m_0(y)\frac{\Phi_t(y)}{\Phi_t(\frac{S^a+S^b}{2})}}{\int m_0(x)\frac{\Phi_t(x)}{\Phi_t(\frac{S^a+S^b}{2})}dx},
\end{align*}
and 
\begin{align*}
    &m_0(y)\frac{\Phi_t(y)}{\Phi_t(\frac{S^a+S^b}{2})}=m_0(y)e^{-2\lambda_0e^{-a\frac{S^a-S^b}{2}}t(\cosh(y-\frac{S^a+S^b}{2})-1)}=m_0(y)e^{-\frac{t}{t_1}(\cosh(y-\frac{S^a+S^b}{2})-1)}
\end{align*}
where $t_1 = \frac{e^{a\frac{S^a-S^b}{2}}}{2\lambda_0}$. We split the integral as
\begin{align*}
    &\int m_0(x)\frac{\Phi_t(x)}{\Phi_t(\frac{S^a+S^b}{2})}dx=\int_{-\infty}^{\frac{S^a+S^b}{2}} m_0(x)\frac{\Phi_t(x)}{\Phi_t(\frac{S^a+S^b}{2})}dx+\int_{\frac{S^a+S^b}{2}}^{+\infty} m_0(x)\frac{\Phi_t(x)}{\Phi_t(\frac{S^a+S^b}{2})}dx
\end{align*}
and we study the second term. We get
\begin{align*}
    \int_{\frac{S^a+S^b}{2}}^{+\infty} m_0(x)\frac{\Phi_t(x)}{\Phi_t(\frac{S^a+S^b}{2})}dx&=\frac{t_1}{t}\int_0^{+\infty}m_0(\frac{S^a+S^b}{2}+\text{arccosh}(1+u\frac{t_1}{t}))\frac{e^{-u}}{\sqrt{(1+u\frac{t_1}{t})^2-1}}du\\
    &=\sqrt{\frac{t_1}{t}}\int_0^{+\infty}m_0(\frac{S^a+S^b}{2}+\text{arccosh}(1+u\frac{t_1}{t}))\frac{e^{-u}}{\sqrt{2u+u^2\frac{t_1}{t}}}du\\
    &\sim\sqrt{\frac{t_1}{t}}m_0(\frac{S^a+S^b}{2})\int_0^{+\infty}\frac{e^{-u}}{\sqrt{2u}}\\
    &=\sqrt{\frac{t_1}{t}}m_0(\frac{S^a+S^b}{2})\frac{\sqrt{\pi}}{2}
\end{align*}
and we conclude that
\begin{align*}
    u_t(y)\sim \sqrt{\frac{t}{\pi t_1}}\frac{m_0(y)}{m_0(\frac{S^a+S^b}{2})}e^{-\frac{t}{t_1}(\cosh(y-\frac{S^a+S^b}{2})-1)}.
\end{align*}

\section{Proof of Proposition \ref{prop::slow}}\label{app::proof_slow}
The system \eqref{sys2} can be rewritten as
\begin{align}\label{sys_dem}
    \begin{split}
    \frac{1}{t_1}(u(t)+tu'(t))&=2\beta\frac{1}{v(t)+u(t)} +\frac{1}{t_1},\\
    \frac{1}{t_1}(v(t)+tv'(t))&=-2\beta\frac{1}{v(t)+u(t)} +\frac{1}{t_1},
\end{split}
\end{align}
so 
\begin{align*}
    (tv(t))'=(tu(t))'+2.
\end{align*}
This implies that
\begin{align*}
    v(t)=-u(t)+2+\frac{K}{t}
\end{align*}
for some constant $K$. Plugging this into \ref{sys_dem} and integrating we get
\begin{align*}
    u(t)t=&t(1+\beta t_1)-\beta t_1 \frac{K}{2}\log(t+\frac{K}{2})+R\\
    v(t)t=&t(1-\beta t_1)+\beta t_1 \frac{K}{2}\log(t+\frac{K}{2})+R'
\end{align*}
for constant $R,R'$. This complies with the previous inequality only for $K=R=R'=0$.
\end{document}